\documentclass[11pt,letterpaper]{article}
\usepackage{mathrsfs}
\usepackage{geometry}
\usepackage{graphicx}
\usepackage{amssymb}
\usepackage{epstopdf}
\usepackage{amsmath,amscd}
\usepackage{amsthm}
\theoremstyle{plain} \textwidth=430pt \textheight=660pt
\makeatletter

\newcommand{\Rmnum}[1]{\expandafter\@slowroman #1@}
\makeatother

\newtheorem{theorem}{Theorem}[section]

\newtheorem{lemma}[theorem]{Lemma}
\newtheorem{proposition}[theorem]{Proposition}
\newtheorem{corollary}[theorem]{Corollary}
\newtheorem{example}[theorem]{Example}

\title{Critical Temperature of Periodic Ising Models}
\author{Zhongyang Li\footnote{Department of Pure Mathematics and Mathematical Statistics, University of Cambridge, Cambridge, GB30WA, UK, zl296@statslab.cam.ac.uk}}
\date{}

\begin{document}

\maketitle
\begin{abstract}
A periodic Ising model has interactions which are
invariant under translations of a full-rank sublattice
$\mathfrak{L}$ of $\mathbb{Z}^2$. We prove an exact, quantitative
characterization of the critical temperature, defined as the supremum of
temperatures for which the spontaneous magnetization is strictly positive. 
For the ferromagnetic model, the critical temperature is the solution of a
certain algebraic equation, resulting from the condition that the spectral
curve of the corresponding dimer model on the Fisher graph has a
real zero on the unit torus. With our technique we provide a simple proof
for the exponential decay of spin-spin correlations above the critical temperature, as well as the
exponential decay of the edge-edge correlations for all non-critical
edge weights of the corresponding dimer model on periodic Fisher graphs.
\end{abstract}

\section{Introduction}

Phase transition is a central topic
of statistical mechanics. A system undergoes a phase transition
whenever, for some value of the temperature and other relevant
thermodynamic parameters, two or more phases can coexist in
equilibrium. Different properties of the pure phases manifest
themselves also as discontinuities in certain observables as a
function of the appropriate thermodynamic variables, e.g.
discontinuity of the magnetization as a function of the magnetic
field in a ferromagnet. Of special interest is the temperature at
which the phase transition occurs, that is, the
\textbf{critical temperature}.

Lebowitz and Martin-L\"{o}f \cite{lm} defined the critical
temperature, $T_c$, for Ising ferromagnets as the supremum of
temperatures such that the spontaneous magnetization is strictly positive.  Lebowitz \cite{le} identified $T_c$ with the self-dual point for an isotropic, two-dimensional
square grid Ising model.  Based on various
correlation inequalities and differential inequalities,
 Aizenman, Bursky and Fern\'{a}ndez \cite{abf} characterize the
phase transition by the exponential decay of two-point spin
correlations above $T_c$, for ferromagnetic Ising models with dimension $D\geq 2$. 

For the two-dimensional Ising model, another approach is to
apply the Fisher correspondence \cite{fi}, which is a
measure-preserving bijection between the even spin-spin correlation
functions of the Ising model on a graph $G$, and the edge
probabilities of a dimer model on a decorated graph, the
\textbf{Fisher graph}. Since then, dimer techniques have been a
powerful tool in solving the two-dimensional Ising model,
see for example the paper of Kasteleyn \cite{ka2}, and the book of
McCoy and Wu \cite{mw}. However, due to the complexity related to
large matrices, the two-dimensional Ising model has acquired a
notorious reputation for difficulty. On the mathematical side,
Kenyon \cite{ke2,kos,ko} has proven numerous spectacular results about the dimer model
on bipartite graphs in recent years. Those results make it possible
to look at the two-dimensional Ising model from a new perspective.
In our paper, we follow the dimer approach to the Ising model, and
explicitly characterize the critical temperature, defined as in
\cite{lm} and \cite{abf}, with the zero of an algebraic equation.

\begin{theorem}Let $T_c$ be the critical temperature of the ferromagnetic, two-dimensional periodic Ising model,
defined by the largest temperature such that the spontaneous
magnetization is strictly positive, then $0<T_c<\infty$. $T_c$ is determined by the condition that the spectral curve of the corresponding
dimer model on the Fisher graph has a real zero on $\mathbb{T}^2$.
\end{theorem}

Theorem 1.1 allows one to compute the critical temperature
of arbitrary periodic Ising ferromagnets accurately, by solving an
algebraic equation, see Example~6.14. The proof of Theorem 1.1 consists of
3 steps:

\textbf{Step 1}: Applying Lebowitz's technique \cite{le2} and the
FKG inequality, we prove that the weak
limit of the even spin-spin correlation functions is independent of
the boundary conditions as the size of the graph goes to infinity.
For the uniqueness theorem about the Gibbs measure of dimer models on a more general class of non-bipartite graphs,
we refer to Corollary B.1.

\textbf{Step 2}: Using an $n\times n$ torus to approximate the
infinite bi-periodic graph, we express the two-point spin
correlation functions as the determinant of certain block Toeplitz
matrices. For basics of Toeplitz and Hankel matrices, we refer to the appendix.  We prove that the determinant of the symbol of such block
Toeplitz matrices is identically~1. Using the FKG inequality, we
prove that there is a unique $T_{c,p}$, defined to be the lowest
temperature such that the limit of the two-point spin correlation
function is 0. The analytic property for the limit of two-point spin
correlation function in case the entries of the block Toeplitz matrices
are analytic with respect to the reciprocal temperature follows from
Widom's theorem \cite{wi} and a result (Lemma 4.6) from operator analysis
{\cite{gk}}. Hence $T_{c,p}$ has to satisfy the condition that the
spectral curve has a real zero on $\mathbb{T}^2$.

\textbf{Step 3}: Characterization of $T_c$ and the condition that the
spectral curve has a real zero on $\mathbb{T}^2$ (critical dimer
weights). Since the two-point spin correlation functions are independent of 
different translation invariant Gibbs measures, we derive
$T_{c,p}=T_{c}$, and at the critical temperature of the Ising model, the corresponding dimer model is also critical. We also prove
that as $\beta$ (reciprocal temperature) increases from 0 to
$\infty$, there is a unique $\beta_0$ ($0<\beta_0<\infty$), such that
the spectral curve has a real zero on $\mathbb{T}^2$. The uniqueness
of the critical dimer weights implies its identification with the criticality of the Ising model.

Here is a description of the structure of the paper.
Section 2 proves the uniqueness of even spin-spin correlation functions for periodic ferromagnetic 2D Ising models. In section 3 we discuss the correspondence between Ising models on a square grid and dimer models on the Fisher graph. We also prove that the spectral curve is invariant under the duality transformation. In Section 4, we complete Step 2 and Step 3 to prove the main theorem.
The block Toeplitz and the transfer matrix technique also
provides a simple proof of exponential decay of the spin-spin
correlations at high temperature for Ising models with interactions satisfying certain symmetry properties, see
Section 5. For a proof of exponential decay for general
interactions, we refer to \cite{abf,ai2}. In the appendix, we discuss several basic facts of Toeplitz and Hankel matrices for the readers' reference.
\\
\\
\noindent\textbf{Acknowledgements} The author would like to thank Richard Kenyon
for many stimulating discussions. The author is also grateful to the anonymous reviewer for the detailed and highly valuable suggestions. The author acknowledges support from the EPSRC under grant EP/103372X/1.

\section{Translation Invariant Gibbs Measure}

\subsection{F.K.G Inequality for Periodic Interactions}

We consider a Ising spin system with a ferromagnetic pair
interaction in a finite box $\Lambda$ of $\mathbb{Z}^2$, embedded on the whole plane,
i.e, at each point $p$ of the lattice there is a spin
$\sigma_p=\pm1$, and the conditional probability of a spin
configuration in the box $\Lambda$ given a configuration outside it
is proportional to
\begin{eqnarray}
e^{-E_{\Lambda}(\sigma)}=e^{\beta(\frac{1}{2}\sum_{p\neq
q\in\Lambda,|p-q|=1}J_{pq}\sigma_p\sigma_q+\sum_{p\in\Lambda}h\sigma_p+\sum_{p\in\Lambda,q\in\Lambda^c}J_{pq}\sigma_p\sigma_q)},\label{probability}
\end{eqnarray}
where $J_{pq}>0$ is the pair interaction. Assume $J_{pq}$ has period
$(m,n)$, i.e, for any $i,j\in\mathbb{Z}$,
$J_{pq}=J_{p+(im,jn),q+(im,jn)}$. $h$ is the uniform external
magnetic field, and $\beta$ is the reciprocal temperature. A
boundary condition for the box $\Lambda$, is specified by giving a
probability distribution $b_{\Lambda}$ for the configurations
outside $\Lambda$. Let $\Gamma$ be the set of all Ising spin
configurations on $\Lambda$. Under the partial order relation
``$-$''$<$``$+$'', the partial order set $\Gamma$ is a
\textbf{lattice}, since any two configurations $x$ and $y$ in
$\Gamma$ have a least upper bound $x\vee y$ and a greatest lower
bound $x\wedge y$. Moreover, $\Gamma$ is a \textbf{finite
distributive lattice}, see \cite{fkg} for a definition.  We have the
following theorem:
\begin{theorem}
Let $f$ be an increasing function with respect to the partial order
relation on $\Gamma$, $\mu_1,\mu_2$ be two probability measures, if
all the following conditions are satisfied
\begin{itemize}
\item the boundary condition $\delta_1\geq\delta_2$;
\item the external magnetic field $h_1\geq h_2$;
\item the pair interaction $J_{pq}^{1}\geq J_{pq}^2$;
\end{itemize}
Then the expected values of f satisfy
\begin{eqnarray*}
\langle f\rangle_{\mu_1}\geq \langle f\rangle_{\mu_2}.
\end{eqnarray*}
\end{theorem}
\begin{proof}
Since $\Gamma$ is a finite distributive lattice, according to
\cite{ho}, if any spin configuration $A,B\in\Gamma$ satisfy
\begin{eqnarray}
\mu_1(A\vee B)\mu_2(A\wedge
B)\geq\mu_1(A)\mu_2(B)\label{convexcondition},
\end{eqnarray}
then
\begin{eqnarray*}
\langle f\rangle_{\mu_1}\geq \langle f\rangle_{\mu_2}.
\end{eqnarray*}
for any increasing functions $f$ on $\Gamma$.

Let us prove Condition (\ref{convexcondition}) for
$\delta_1\geq\delta_2$, while $h_1=h_2$ and $J_{pq}^1=J_{pq}^2$. If $A$ does not have boundary configuration
$\delta_1$ or $B$ does not have boundary configuration $\delta_2$,
Condition (\ref{convexcondition}) is satisfied automatically since
the right hand side is zero while the left hand side is nonnegative.
Hence to prove Condition (\ref{convexcondition}), it suffices to
prove that for any A with boundary configuration $\delta_1$ and B
with boundary configuration $\delta_2$,
\begin{eqnarray}
\mu(A\vee B)\mu(A\wedge B)\geq\mu(A)\mu(B),\label{convexcondition2}
\end{eqnarray}
where $\mu$ is the probability measure on $\Gamma$ with free
boundary condition. Condition (\ref{convexcondition2}) is true if
energy for all bonds $J_e\geq 0$. Notice that $J_e$ may be
different from edge to edge. It is trivial to check Condition
(\ref{convexcondition2}) for each single bond with 16 different
combinations of spin configurations for A,B. If the inequality is true for the contribution of each edge in the probability, the inequality is true for the probability, since the probability is the product of contributions of each edge, and the partial order relation is defined bond by bond.

The same method can be applied to check Condition (\ref{convexcondition}) spin by
spin or bond by bond for the cases $h_1\geq h_2$ and $J_{pq}^1\geq J_{pq}^2$.
\end{proof}

\noindent\textbf{Remark.} Condition (2) is called the F.K.G inequality.

\subsection{Free Energy} As before, let $\Lambda$ be a finite box of
an infinite periodic square grid. The equilibrium state of the
system in $\Lambda$ is the probability distribution for
configurations in $\Lambda$ defined by (\ref{probability}) together
with boundary condition $b_{\Lambda}$, or equivalently, by the
family of correlation functions
\begin{eqnarray*}
\langle\sigma_A\rangle_{h,\Lambda,b_{\Lambda}}=\langle\prod_{p\in
A}\sigma_p\rangle_{h,\Lambda,b_{\Lambda}},\ \ \ for\ any\
A\subset\Lambda.
\end{eqnarray*}
An equilibrium state of the infinite system is defined to be a
family of correlation functions obtained as the limit of correlation
functions for a sequence of finite boxes with some boundary
conditions
\begin{eqnarray*}
\langle\sigma_A\rangle_{h,b}=\lim_{\Lambda\rightarrow\infty}\langle\sigma_A\rangle_{h,\Lambda,b_{\Lambda}}.
\end{eqnarray*}
We say a sequence of boxes $\{\Lambda_n\}_{n=1}^{\infty}$ tends to
infinity if
\begin{enumerate}
\item each $\Lambda_n$ is connected and $\Lambda_n\subset\Lambda_{n+1}$;
\item $\bigcup_{n}\Lambda_n=\textbf{T}$;
\item for any $l\geq 0$, let $\Lambda^l$ be the set of points in
$\Lambda$ with distance at most $l$ from its boundary, and
$|\Lambda|$ denote the number of points in $\Lambda$, then
\begin{eqnarray*}
\lim_{n\rightarrow\infty}\frac{|\Lambda_n^l|}{|\Lambda_n|}=0\ \
\ (\mathrm{Van\ Hove\ convergence} )
\end{eqnarray*}
\end{enumerate}
The letter $b$ indicates some arbitrary boundary condition, and
$b_{\Lambda}=``+"$ or $b_{\Lambda}=``-"$ will indicate the boundary
conditions defined by putting all spins outside $\Lambda$ $``+"$ or
$``-"$. The statement $\langle\sigma_{A}\rangle_{h,b}$ does not
depend on the boundary condition is equivalent to the statement that
the equilibrium is unique.

For a finite system, the free energy is defined as follows
\begin{eqnarray*}
F(h,\Lambda,b_{\Lambda})=\frac{1}{|\Lambda|}\log(\sum_{\sigma}e^{-E_{
\Lambda}(\sigma)}),
\end{eqnarray*}
where $E_{\Lambda}(\sigma)$ is defined in (\ref{probability}). We
have the following theorem about the free energy
\begin{theorem}
As $\Lambda\rightarrow\infty$, the limit of
$F(h,\Lambda,b_{\Lambda})$ exists and is independent of the boundary
condition. Assume
\begin{eqnarray*}
\lim_{\Lambda\rightarrow\infty}F(h,\Lambda,b_{\Lambda})=F(h).
\end{eqnarray*}
Then $F(h)$ is convex, and analytic for $h\neq 0$
\end{theorem}
\begin{proof}
First of all, for each fixed $\Lambda$ and $b_{\Lambda}$,
$F(h,\Lambda,b_{\Lambda})$ is a convex function of $h$. Assume
\begin{eqnarray*}
A_{\sigma}&=&e^{-E_{\Lambda}(\sigma)}\geq 0,\\
B_{\sigma}&=&\sum_{p\in\Lambda}\beta\sigma_p.
\end{eqnarray*}
Then
\begin{eqnarray*}
\frac{\partial^2F(h,\Lambda,b_{\Lambda})}{\partial h^2}
=\frac{(\sum_{\sigma}A_{\sigma})(\sum_{\sigma}A_{\sigma}B_{\sigma}^2)-(\sum_{\sigma}A_{\sigma}B_{\sigma})^2}{|\Lambda|(\sum_{\sigma}
A_{\sigma})^2}\geq 0.
\end{eqnarray*}

The inequality follows from applying the Cauchy-Schwartz to $\{\sqrt{A_{\sigma}}\}_{\sigma}$ and $\{\sqrt{A_{\sigma}}B_{\sigma}\}_{\sigma}$. Equality  holds only if all the $B_\sigma$ are equal, which is
impossible, hence $\frac{\partial^2F(h,\Lambda,b_{\Lambda})}{\partial h^2}>0$,
$F(h,\Lambda,b_{\Lambda})$ is convex.

We claim that if
$\lim_{\Lambda\rightarrow\infty}F(h,\Lambda,b_{\Lambda})$ exists, it
is independent of the boundary condition $b_{\Lambda}$. Assume 
$G(h,\Lambda)$ is the free energy for the free boundary conditions, namely
\begin{eqnarray*}
G(h,\Lambda)=\frac{1}{|\Lambda|}\log\sum_{\sigma}e^{\beta(\frac{1}{2}\sum_{p\neq
q\in\Lambda,|p-q|=1}J_{pq}\sigma_p\sigma_q+\sum_{p\in\Lambda}h\sigma_p)},
\end{eqnarray*}
then
\begin{eqnarray*}
G(h,\Lambda)-\frac{|\partial\Lambda|}{|\Lambda|}\beta\max_{p,q}J_{pq}\leq
F(h,\Lambda,b_{\Lambda})\leq
G(h,\Lambda)+\frac{|\partial\Lambda|}{|\Lambda|}\beta\max_{p,q}J_{pq}.
\end{eqnarray*}
Under the assumption that $\Lambda\rightarrow\infty$ in the sense of
Van Hove,
\begin{eqnarray*}
\lim_{\Lambda\rightarrow\infty}\frac{|\partial\Lambda|}{|\Lambda|}=0.
\end{eqnarray*}
Since $J_{pq}$ has finite period, we have
\begin{eqnarray*}
\lim_{\Lambda\rightarrow\infty}F(h,\Lambda,b_{\Lambda})=\lim_{\Lambda\rightarrow\infty}G(h,\Lambda),
\end{eqnarray*}
which is independent of boundary conditions.

It is convenient to introduce a new variable
\begin{eqnarray*}
z=e^{-2\beta h},
\end{eqnarray*}
then
\begin{eqnarray*}
e^{-|\Lambda|G}=Q(z)e^{|\Lambda|h\beta}\label{relation},
\end{eqnarray*}
where
\begin{eqnarray}
Q(z)=\sum Q_nz^n.\ \ \ \ (n=0,1,...,|\Lambda|)\label{polynomial}
\end{eqnarray}
The coefficients $Q_n$ are the contribution to the partition
function of the Ising lattice in zero external field from
configurations with the number of $"-"$ spins equal to $n$. The
following lemma is proved by Lee and Yang \cite{ly}
\begin{lemma}
Let $x_{\alpha\beta}=x_{\beta\alpha}(\alpha\neq\beta,
\alpha,\beta=1,2,...,n)$ be real numbers whose absolute values are
less than or equal to 1. Let $V=\{1,2,\cdots,n\}$. Divide the integers $1,2,...,n$ into 2
groups $a$ and $b$ so that there are $\gamma$ integers in group $a$
and $(n-\gamma)$ integers in group $b$. Consider the product of all
$x_{\alpha\beta}$ where $\alpha$ belongs to group $a$ and $\beta$
belongs to group $b$. We shall denote by $Q_{\gamma}$ the sum of all
such products over all the $\frac{n!}{\gamma!(n-\gamma)!}$ possible
ways of dividing the $n$ integers, in other words
\begin{eqnarray*}
Q_{\gamma}=\sum_{a\subset V, |a|=\gamma}\prod_{\alpha\in a}\prod_{\beta\in V\setminus a}x_{\alpha\beta}.
\end{eqnarray*}
Consider the polynomial
\begin{eqnarray*}
Q(z)=1+Q_1z+...+Q_nz^n.
\end{eqnarray*}
Then all the roots of the equation $Q=0$ are on the unit circle.
\end{lemma}
Apply the lemma to (\ref{polynomial}), we obtain that the roots of
$Q(z)=0$ are on the unit circle. The distribution of roots of
$Q(z)=0$ as $\Lambda\rightarrow\infty$ may be described by a measure $d\mu(\theta)$, so that $|\Lambda|d\mu(\theta)$ is the
number of roots between $e^{i\theta}$ and $e^{i(\theta+d\theta)}$.
Since $Q(z)$ has real coefficients, we have
\begin{eqnarray*}
d\mu(\theta)=d\mu(-\theta).
\end{eqnarray*}
Taking the logarithm of (\ref{relation}), and taking the limit as
$\Lambda\rightarrow\infty$, we have
\begin{eqnarray*}
F(h)&=&h\beta+\int_{0}^{2\pi}\log(z-e^{i\theta})d\mu(\theta)\\
&=&h\beta+\int_{0}^{\pi}\log(z^2-2z\cos\theta+1)d\mu(\theta).
\end{eqnarray*}
The singularities of $F(h)$ corresponds to zeros of $z^2-2z+1$.
Since $z=e^{-2\beta h}>0$, the only possible singularity of $F(h)$
happens at $h=0,z=1$. The convexity of $F(h)$ follows from the
convexity of $F(h,\Lambda,b_{\Lambda})$.
\end{proof}

\subsection{High Temperature}

Define the average magnetization as follows
\begin{eqnarray*}
m(h,\Lambda,b_{\Lambda})=\langle\frac{1}{|\Lambda|}\sum_{p\in\Lambda}\sigma_p\rangle_{h,\Lambda,b_{\Lambda}}=\frac{1}{\beta}\frac{\partial
F(h,\Lambda,b_{\Lambda})}{\partial h}.
\end{eqnarray*}
We have the following lemma about the average magnetization
\begin{lemma}
As $\Lambda\rightarrow\infty$, the limit of $m(h,\Lambda,b_{\Lambda})$ exists and is independent of boundary conditions. Moreover,
\begin{eqnarray*}
\lim_{\Lambda\rightarrow\infty}m(h,\Lambda,b_{\Lambda})=\frac{dF(h)}{dh},\
\ \ \ \ \forall h\neq 0.
\end{eqnarray*}
\end{lemma}
\begin{proof}
$\frac{1}{|\Lambda|}\sum_{p\in\Lambda}\sigma_p$ is an increasing
function on $\Gamma_\Lambda$, and the F.K.G inequality says that
\begin{eqnarray}\label{boundaryrelation}
m(h,\Lambda,-)\leq m(h,\Lambda,b_\Lambda)\leq m(h,\Lambda,+).
\end{eqnarray}
For any $\Lambda'\supseteq\Lambda$, $m(h,\Lambda,+)$ can be obtained
from $m(h,\Lambda',+)$ by adding an infinite positive external
magnetic field on $\Lambda'\setminus \Lambda$. By the F.K.G inequality we
have
\begin{eqnarray*}
m(h,\Lambda',+)\leq m(h,\Lambda,+).
\end{eqnarray*}
Hence $m(h,\Lambda,+)$ is decreasing as $\Lambda$ increases.
Therefore $\lim_{\Lambda\rightarrow\infty}m(h,\Lambda,+)$ exists,
denoted by $m(h,+)$, for all $h$. Moreover, since
$|m(h,\Lambda,b_{\Lambda})|\leq 1$, according to the Dominated
Convergence Theorem, we have
\begin{eqnarray*}
\int_{h_0}^{h}m(\hbar,+)d\hbar=\lim_{\Lambda\rightarrow\infty}\int_{h_0}^{h}m(\hbar,\Lambda,+)d\hbar=F(h)-F(h_0),
\end{eqnarray*}
where $h>h_0>0$ or $0>h>h_0$. Since $F(h)$ is analytic in h when
$h\neq 0$, we have
\begin{eqnarray*}
m(h,+)=\frac{dF(h)}{dh}.
\end{eqnarray*}
Similar process shows that
\begin{eqnarray*}
m(h,-)=\frac{dF(h)}{dh}.
\end{eqnarray*}
And the lemma follows from \ref{boundaryrelation}.
\end{proof}

Since $F(h)$ is analytic in $h$ when $h\neq 0$, we define
\begin{eqnarray*}
m^{*}=\lim_{h\rightarrow 0+}\frac{d F(h)}{d h}.
\end{eqnarray*}
Applying F.K.G inequality and follow exactly the same process by
Lebowitz and Martin-L\"{o}f \cite{lm}, we have the following lemma:
\begin{lemma}
When the external magnetic field is zero there is a unique
equilibrium state for the periodic, ferromagnetic infinite system if
and only if $m^*=0$.
\end{lemma}

$\frac{1}{|\Lambda|}\sum_{p\in\Lambda}\sigma_p$ is an increasing
function on $\Gamma$, according to the F.K.G inequality,
$m(h,\Lambda,b_{\Lambda})=\langle\frac{1}{|\Lambda|}\sum_{p\in\Lambda}\sigma_p\rangle_{\Lambda,b_{\Lambda}}$
is an increasing function in $J_{i,j}$ and $h$, hence if we fix
$J_{i,j}$ and $h$, it is an increasing function in $\beta$. When
$\beta=0$, we have a uniform distribution for all configurations,
therefore
\begin{eqnarray*}
m(h,\Lambda,b_{\Lambda})|_{\beta=0}=0.
\end{eqnarray*}
As a result $m(h,\Lambda,b_{\Lambda})\geq0$ for any $\beta$. As its
limit $m^{*}$ is nonnegative and increasing in $\beta$. The critical
temperature $\beta_c$ is uniquely defined by the conditions
$m^{*}(\beta)=0$ for $\beta<\beta_c$ and $m^{*}(T)>0$ for
$\beta>\beta_c$. Then we have the following theorem
\begin{theorem}
When $\beta<\beta_c$, there exists a unique probability measure of
ferromagnetic, periodic Ising spin system without external magnetic
field.
\end{theorem}

\subsection{Uniqueness of Spin-Spin Even Correlations}

This section is devoted to the proof of Theorem 2.7. The proof is divided into proving several lemmas, as specified below.

\begin{theorem} Given a bi-periodic, ferromagnetic Ising model with fixed reciprocal temperature $\beta$, periodic interactions $J$, and external magnetic field $h=0$, the spin-spin even correlation functions are unique under any translation invariant Gibbs measure. That is, fix a finite subset $A$ of the vertices of the graph such that $|A|$ is even, for any translation invariant Gibbs measure $\mu$, $\langle\sigma_A\rangle_{\mu}$ is independent of $\mu$.
\end{theorem}

Define
\begin{eqnarray*}
\rho_A=\prod_{v\in A}\frac{\sigma_v+1}{2}.
\end{eqnarray*}

\begin{lemma}$\langle\rho_A\rangle_{\pm}=\lim_{\Lambda\rightarrow\infty}\langle\rho_A\rangle_{\Lambda,\pm}$
exists and is translation invariant. That is
$\langle\rho_{A+g}\rangle_{\pm}=\langle\rho_{A}\rangle_{\pm}$, where
$g$ is a translation vector.
\end{lemma}
\begin{proof} The F.K.G inequality applies to $f=\rho_A$, because
this is an increasing function in each spin variable. This implies
that $\langle\rho_A\rangle_{\Lambda',+}\leq
\langle\rho_A\rangle_{\Lambda,+}$, if $\Lambda'\supseteq\Lambda$, by
the F.K.G inequality. Hence
\[\langle\rho_A\rangle_{+}=\lim_{\Lambda\rightarrow\infty}\langle\rho_A\rangle_{\Lambda,+}\]
exists and $\langle\rho_A\rangle_{+}\leq
\langle\rho_A\rangle_{\Lambda,+}$. The translation invariance of the
$\langle\rho_A\rangle_{\pm}$ follows from the uniqueness of taking
the limit $\Lambda\rightarrow\infty$.
\end{proof}

\begin{lemma} \[\langle\sigma_A\rangle_{\Lambda,b_{\Lambda}}\leq \langle\sigma_A\rangle_{\Lambda,+}.\]
\end{lemma}
\begin{proof} Let $k$ denote any edge connecting one vertex in $\Lambda$ and one vertex outside $\Lambda$. Consider
\begin{eqnarray*}
\langle\sigma_A\rangle=\frac{\sum_{\sigma}\sigma_A e^{\sum_{i\sim
j,i,j \in \Lambda^{0}}J_{ij}\sigma_i\sigma_j+\sum_{k\in\partial
\Lambda}J_k\sigma_k}}{Z},\\
\langle\sigma_A\rangle'=\frac{\sum_{\sigma}\sigma_A e^{\sum_{i\sim
j}J_{ij}\sigma_i\sigma_j+\sum_{k\in\partial
\Lambda}J_k'\sigma_k}}{Z'},
\end{eqnarray*}
where
\begin{eqnarray*}
Z=\sum_{\sigma}e^{\sum_{i\sim j,i,j \in
\Lambda^{0}}J_{ij}\sigma_i\sigma_j+\sum_{k\in\partial
\Lambda}J_k\sigma_k},\\
Z'=\sum_{\sigma}e^{\sum_{i\sim j,i,j \in
\Lambda^{0}}J_{ij}\sigma_i\sigma_j+\sum_{k\in\partial
\Lambda}J'_k\sigma_k},
\end{eqnarray*}
 then
\begin{eqnarray}
\langle\sigma_A\rangle-\langle\sigma_A\rangle'&=&\frac{\sum_{\sigma,\sigma'}(\sigma_A-\sigma_A')e^{\sum_{i\sim
j,i,j\in\Lambda^{0}}J_{ij}(\sigma_i\sigma_j+\sigma_i'\sigma_j')+\sum_{k\in\partial\Lambda}(J_k\sigma_k+J_k'\sigma_k')}}{ZZ'}\\
&=&\frac{\sum_t(1-t_A)\sum_{\sigma}\sigma_Ae^{\sum_{i\sim
j,i,j\in\Lambda^0}J_{ij}(1+t_{i}t_{j})\sigma_i\sigma_j+\sum_{k\in\partial\Lambda}(J_k+t_kJ_k')\sigma_k}}{ZZ'},\label{correlationdifference}
\end{eqnarray}
where $t_i=\sigma_i\sigma'_i$. Here $J_k'=\pm J_k$ reflects the boundary conditions. It follows from the ferromagnetic
condition that when
\[J_{k}\geq |J'_{k}|,\qquad \forall k\in\partial\Lambda,\]
the right side of (\ref{correlationdifference}) is nonnegative.
\end{proof} 

Let $A=B\vartriangle C=B\cup C\setminus B\cap C$ be the symmetric
difference between the sets $B$ and $C$, then $t_A=t_Bt_C=\pm 1$,
and
\begin{eqnarray}
1-t_Bt_C=\pm(t_B-t_C).\label{partialinequality}
\end{eqnarray}
Substituting (\ref{partialinequality}) on the right of
(\ref{correlationdifference}) and going back to the $\sigma'$
variables, we obtain our basic inequality for
$\sigma_{A}=\sigma_{B}\sigma_{C}$:
\begin{eqnarray}
\langle\sigma_B\sigma_C\rangle-\langle\sigma_B\sigma_C\rangle'\geq|\langle\sigma_B\rangle\langle\sigma_C\rangle'-\langle\sigma_B\rangle'\langle\sigma_C\rangle|.\label{spininequality}
\end{eqnarray}
 Then absolute value sign on the right side of (\ref{spininequality}) follows from the fact that the right side of (\ref{correlationdifference}) is always nonnegative given ferromagnetic interactions. Namely, $Z, Z'$ are partition functions of Ising models, which are positive. For any give $t_i$ at all vertices, we consider an Ising model with interaction constants
 \begin{eqnarray*}
 \tilde{J}_{ij}&=&J_{ij}(1+t_{i}t_{j})\geq 0,\qquad \forall i\sim j.\\
 \tilde{J}_{k}&=&J_k+t_kJ_k'\geq 0,\qquad \mathrm{if\ J_k\geq|J_k'|}\  \forall k\in\partial\Lambda.
 \end{eqnarray*}
 Since all interaction constants are nonnegative, we have a ferromagnetic Ising model, as a result,
 \begin{eqnarray*}
 \langle\sigma_A\rangle \tilde{Z}= \sigma_Ae^{\sum_{i\sim
j,i,j\in\Lambda^0}J_{ij}(1+t_{i}t_{j})\sigma_i\sigma_j+\sum_{k\in\partial\Lambda}(J_k+t_kJ_k')\sigma_k}\geq 0,
 \end{eqnarray*}
 where $\tilde{Z}$ is the partition function for the new Ising model. Moreover $1-t_A$ is always nonnegative, that is the reason that the right side of $(\ref{correlationdifference})$ is nonnegative, and (\ref{spininequality}) follows.

\begin{lemma}
\begin{eqnarray}
\lim_{\beta'\rightarrow \beta
+}\langle\rho_A\rangle_{\beta',+}=\langle\rho_A\rangle_{\beta,+}\label{rightcontinuity},\\
\lim_{\beta'\rightarrow \beta
-}\langle\rho_A\rangle_{\beta',-}=\langle\rho_A\rangle_{\beta,-}\label{leftcontinuity}.
\end{eqnarray}
\begin{proof}Since $\langle\rho_A\rangle_{\beta,+}\leq\langle\rho_A\rangle_{\beta,\Lambda,+},$
\[\lim_{\beta'\rightarrow\beta+}\langle\rho_A\rangle_{\beta',+}\leq\lim_{\beta'\rightarrow\beta+}\langle\rho_A\rangle_{\beta',\Lambda,+}=\langle\rho_A\rangle_{\beta,\Lambda,+}.\]
Letting $\Lambda\rightarrow\infty$, we have
\[\lim_{\beta'\rightarrow\beta+}\langle\rho_A\rangle_{\beta',+}\leq\langle\rho_A\rangle_{\beta,+}.\]
But F.K.G inequality implies that
\[\langle\rho_A\rangle_{\beta,+}\leq \langle\rho_A\rangle_{\beta',+},\] for $\beta\leq
\beta'$, so
\[\langle\rho_A\rangle_{\beta,+}\leq\lim_{\beta'\rightarrow\beta+}\langle\rho_A\rangle_{\beta',+}.\]
Then (\ref{rightcontinuity}) is proved, and (\ref{leftcontinuity})
can be proved analogously.
\end{proof}
\end{lemma}

\begin{lemma} $\frac{\partial F}{\partial \beta}$ is continuous for all
the ferromagnetic interactions.
\end{lemma}
\begin{proof}First of all, notice that $\langle\sigma_i\sigma_j\rangle_{+}=\langle\sigma_i\sigma_j\rangle_{-}$
by symmetry. Let $\frac{\partial F}{\partial \beta^{+}}$$\left(\frac{\partial F}{\partial \beta^{-}}\right)$ be the right(left) derivative of $F$ with respect to $\beta$.  The assertion will follow if we prove that
\begin{eqnarray}
\frac{\partial F}{\partial
\beta^{\pm}}=\lim_{\Lambda\rightarrow\infty}\frac{1}{|\Lambda|}\sum_{i\sim
j,i,j\in\Lambda}\langle\sigma_i\sigma_j\rangle_{\Lambda,\pm}=\frac{1}{|\Lambda_1|}\sum_{i\sim
j,i,j\in\Lambda_1}\langle\sigma_i\sigma_j\rangle_{\pm},
\label{freeenergyderivative}
\end{eqnarray}
where $\Lambda_1$ is the quotient graph of the infinite periodic
graph with respect to translation. Now we prove
(\ref{freeenergyderivative}). Consider first the boundary condition
$+$. For any $\epsilon>0$, let $\Lambda_{\epsilon}$ be a box
containing all the edges of $\Lambda_1$, such that
$\langle\sigma_i\sigma_j\rangle_{\Lambda_\epsilon,+}\leq
\langle\sigma_i\sigma_j\rangle_{+}+\epsilon$, for all $i\sim j$,
$i,j\in\Lambda_1$. Then for any translation vector $g$, and
$\Lambda$ such that $\Lambda_\epsilon+g\subseteq\Lambda$, we have
for all $i\sim j, i,j\in \Lambda_0$,
\begin{eqnarray}
\langle\sigma_i\sigma_j\rangle_{+}=\langle\sigma_{i+g}\sigma_{j+g}\rangle_{+}\leq \langle\sigma_{i+g}\sigma_{j+g}\rangle_{\Lambda,+}\nonumber\\
\leq\langle\sigma_{i+g}\sigma_{j+g}\rangle_{\Lambda_{\epsilon}+g,+}=\langle\sigma_i\sigma_j\rangle_{\Lambda_{\epsilon},+}\leq\langle\sigma_i\sigma_j\rangle_{+}+\epsilon.
\label{comparecorrelation}
\end{eqnarray}
Using the decomposition
\begin{eqnarray*}
\frac{1}{|\Lambda|}\sum_{i\sim
j,i,j\in\Lambda}\langle\sigma_i\sigma_j\rangle_{\Lambda,\pm}=\frac{1}{|\Lambda|}\sum_{i,j\in\Lambda_1,i\sim
j}\sum_{i+g\in\Lambda,j+g\in\Lambda}\langle\sigma_{i+g}\sigma_{j+g}\rangle_{\Lambda,\pm},
\end{eqnarray*}
where the second sum is over $g$, $\lim_{\Lambda\rightarrow\infty}\frac{|\partial\Lambda|}{|\Lambda|}=0$
and (\ref{comparecorrelation}) imply that
\begin{eqnarray*}
\frac{1}{|\Lambda_1|}\sum_{i,j\in\Lambda_1,i\sim
j}\langle\sigma_i\sigma_j\rangle_{+}\leq\liminf\frac{1}{|\Lambda|}\sum_{i\sim
j}\langle\sigma_i\sigma_j\rangle_{\Lambda,+}\\
\leq\limsup \frac{1}{|\Lambda|}\sum_{i\sim
j}\langle\sigma_i\sigma_j\rangle_{\Lambda,+}\leq\frac{1}{|\Lambda_1|}\sum_{i,j\in\Lambda_1,i\sim
j}\langle\sigma_i\sigma_j\rangle_{+}+\epsilon,
\end{eqnarray*}
for any $\epsilon$, which means that
\[\lim_{\Lambda\rightarrow\infty}\frac{1}{|\Lambda|}\sum_{i\sim j,i,j\in\Lambda}\langle\sigma_i\sigma_j\rangle_{\Lambda,+}=\lim_{\Lambda\rightarrow\infty}\frac{\partial F(\Lambda,+)}{\partial \beta}=\frac{1}{|\Lambda_1|}\sum_{i,j\in\Lambda_1,i\sim j}\langle\sigma_i\sigma_j\rangle_+,\]
and similarly for the boundary condition ``$-$''. Since $F$ is
convex with respect to $\beta$, the discontinuous points of  $\frac{\partial{F}}{\partial\beta}$
is countably many $\beta$'s, at most. For any finite box $\Lambda$ of a
locally finite periodic graph
\[\left|\frac{1}{|\Lambda|}\sum_{i\sim j,i,j\in\Lambda}\langle\sigma_i\sigma_j\rangle\right|\leq \max_{v\in \Lambda_1}\deg(v).\]
According to the Dominated Convergence Theorem,
\begin{eqnarray*}
\int_{\beta_0}^{\beta_1}\lim_{\Lambda\rightarrow\infty}\frac{\partial
F(\Lambda,+)}{\partial
\beta}d\beta=\lim_{\Lambda\rightarrow\infty}\int_{\beta_0}^{\beta_1}\frac{\partial
F(\Lambda,+)}{\partial \beta}d\beta\\
=\lim_{\Lambda\rightarrow\infty}[F(\Lambda,\beta_1,+)-F(\Lambda,\beta_0,+)]=F(\beta_1)-F(\beta_0).
\end{eqnarray*}
The last equality follows from the fact that $F$ is independent of
boundary conditions. Let $\beta'\rightarrow\beta$  from the right side $(\beta'>\beta)$.  We have
\begin{eqnarray*}
\frac{\partial F}{\partial \beta^{+}}=\lim_{\beta'\rightarrow\beta
+}\left.\frac{\partial F}{\partial
\beta}\right|_{\beta=\beta'}=\lim_{\beta'\rightarrow\beta
+}\frac{1}{|\Lambda_1|}\sum_{i\sim
j,i,j\in\Lambda_1}\langle\sigma_i\sigma_j\rangle_{\beta',+}=\frac{1}{|\Lambda_1|}\sum_{i\sim
j,i,j\in\Lambda_1}\langle\sigma_i\sigma_j\rangle_{\beta,+}.
\end{eqnarray*}
So (\ref{freeenergyderivative}) is proved for the boundary condition
$+$, and $-$ is treated analogously.
\end{proof}

\begin{lemma} $\langle\sigma_i\sigma_j\rangle_{\mu}=\langle\sigma_i\sigma_j\rangle_{+}$
for any translation invariant limit probability measure $\mu$ and
any $i\sim j$, if $\frac{\partial F}{\partial \beta}$ is continuous.
\end{lemma}
\begin{proof} Let $\mathbf{K}$ be an arbitrary translation
invariant interaction, define
\[A_{\mathbf{K}}=\frac{1}{|\Lambda_1|}\sum_{i\sim j,i,j\in\Lambda_1}K_{ij}\sigma_i\sigma_j.\]
Then for any translation invariant probability measure $\mu$
\begin{eqnarray}
F(\mathbf{J}+\mathbf{K})\geq F(\mathbf{J})+\langle A_{\mathbf{K}}\rangle_{\mu}.
\label{variationalprinciple}
\end{eqnarray}
That is $\langle A_{\mathbf{K}}\rangle_{\mu}$ is a tangent vector of $F$ at
$\mathbf{J}$ along the direction $\mathbf{K}$. In other words, the expected value of $A_{\mathbf{K}}$ under any translation invariant probability measure corresponds to a tangent vector. For the proof of
(\ref{variationalprinciple}), see \cite{lan}. If
$\left.\frac{\partial F(\mathbf{J}+t\mathbf{J})}{\partial
t}\right|_{t=0}$ is continuous, the tangent vector is unique and its slope is equal to the corresponding derivative, which means the expected values of $ A_{\mathbf{K}}$ under any translation invariant probability measure are the same. Hence if $\frac{\partial F}{\partial\beta}$ is continuous, we have
\[\sum_{i\sim j,i,j\in \Lambda_1}J_{ij}\langle\sigma_i\sigma_j\rangle_{\mu}=\sum_{i\sim j,i,j\in \Lambda_1}J_{ij}\langle\sigma_i\sigma_j\rangle_{+}.\]
Given Lemma 2.9, under the assumption that $J_{ij}>0$, we have
$\langle\sigma_{i}\sigma_{j}\rangle_{+}=\langle\sigma_{i}\sigma_{j}\rangle_{\mu}$ for
all translation invariant limit measure $\mu$.
\end{proof}

\begin{lemma} Assume $J_{ij}>0$, if $\langle\sigma_i\sigma_j\rangle=\langle\sigma_i\sigma_j\rangle'\neq
0$, for all $i\sim j$, then
$\langle\sigma_{E}\rangle=\langle\sigma_{E}\rangle'$, for all sets
$E$ containing an even number of sites.
\end{lemma}
\begin{proof}Any translation invariant Gibbs measure can be considered as the weak limit of Boltzman measures on finite graphs with given boundary conditions, as the size of the graph goes to infinity.  Any boundary conditions can be transformed to $``+"$ boundary conditions by changing the coupling constant on the edges incident to an boundary site. If the configuration at the boundary site is $``+"$, then the new coupling constant on incident edges is unchanged; otherwise the new coupling constant is $-J_e$. Given $J_{ij}>0$, we have $J_{ij}\geq |J_{ij}|$, we derive that the inequality (\ref{spininequality}) is true. 
 Since $\langle\sigma_B\sigma_C\rangle-\langle\sigma_B\sigma_C\rangle'\geq
|\langle\sigma_B\rangle\langle\sigma_C\rangle'-\langle\sigma_B\rangle'\langle\sigma_C\rangle|$,
let $A=B\triangle C$, we have
\[\langle\sigma_A\rangle-\langle\sigma_A\rangle'\geq|\langle\sigma_B\rangle\langle\sigma_A\sigma_B\rangle'-\langle\sigma_B\rangle'\langle\sigma_A\sigma_B\rangle|.\]
Then $\langle\sigma_A\rangle=\langle\sigma_A\rangle'$,
$\langle\sigma_B\rangle=\langle\sigma_B\rangle'\neq 0$ implies
$\langle\sigma_A\sigma_B\rangle=\langle\sigma_A\sigma_B\rangle'$.
The result follows by induction.
\end{proof}

\noindent \textbf{Proof of Theorem 2.7} From Lemma 2.13, it suffices to
prove that $\langle\sigma_i\sigma_j\rangle\neq 0$ and is independent
of boundary conditions for all $i\sim j$. From Lemma 2.12, this is
true if and only if $\frac{\partial F}{\partial \beta}$ is
continuous, and $\frac{\partial F}{\partial \beta}$ is always
continuous, given lemma 2.11.$\Box$
\\
\\
The uniqueness of the spin-spin even correlation functions can lead to the uniqueness of translation invariant Gibbs measures of dimer models on a large class of non-bipartite graphs. Interested readers may look at the appendix for more about this topic. 
\section{Fisher Correspondence}

\subsection{Dimer Model}
The \textbf{Fisher graph} we consider in this paper is a graph
obtained from a honeycomb lattice by replacing each vertex by a
triangle, as illustrated in Figure 2. There are three types of
non-triangle edges on such a graph with different direction, namely
$a$-type, $b$-type, and $c$-type. We give weight 1 on all the $a$
edges(corresponding to horizontal edges in Figure 1). The weights of
$b$-edges and $c$-edges are strictly less than 1 and assigned
periodically with arbitrary period $m\times n$. As we will see, the
dimer model of such a Fisher graph corresponds to the ferromagnetic
Ising model on a periodic square grid. 

\begin{figure}[htbp]
  \centering
\includegraphics{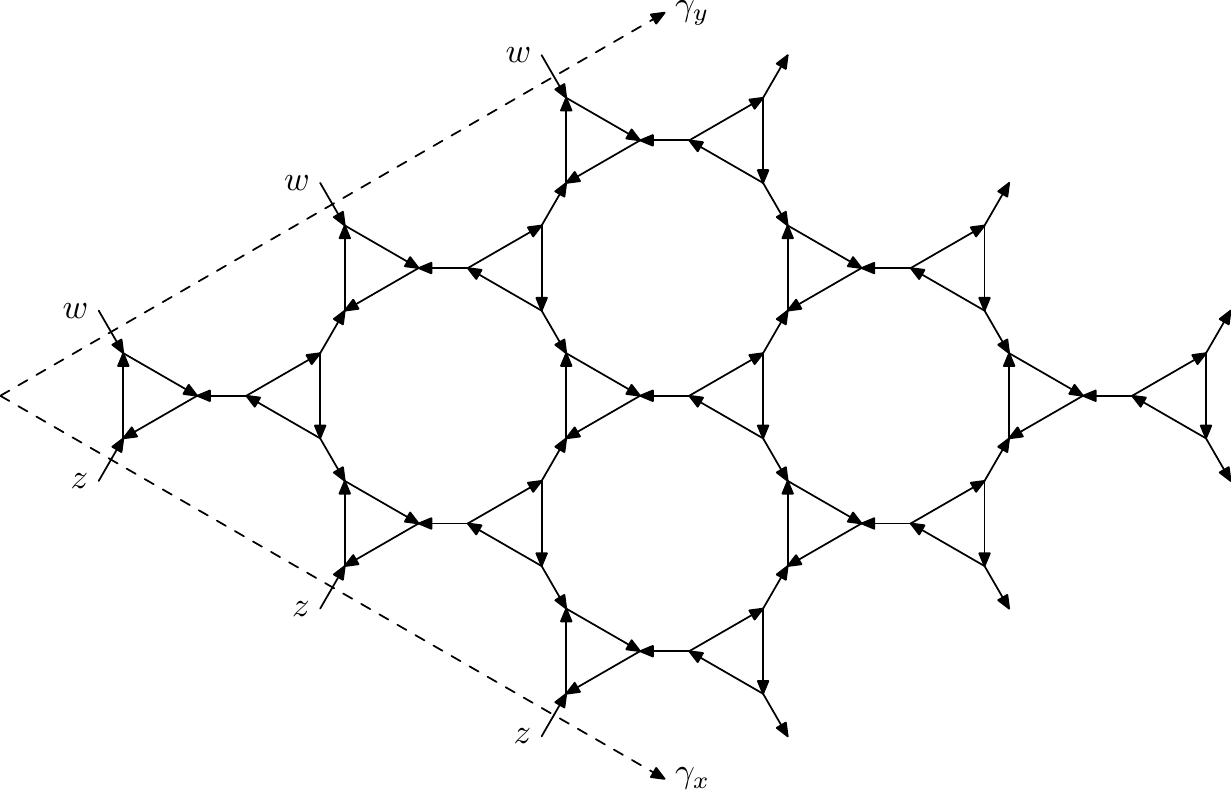}
   \caption{Fisher Graph}
\end{figure}

A \textbf{perfect matching}, or a \textbf{dimer cover}, of a graph
is a collection of edges with the property that each vertex is
incident to exactly one edge. For general results of dimer models, see the Appendix.

A \textbf{perfect matching}, or a \textbf{dimer cover}, of a graph
is a collection of edges with the property that each vertex is
incident to exactly one edge.

To a weighted finite graph $G=(V,E,W)$, the weight $W:
E\rightarrow\mathbb{R}^{+}$ is a function from the set of edges to
positive real numbers. We define a probability measure, called the
\textbf{Boltzmann measure} $\mu$ with sample space the set of dimer
covers. Namely, for a dimer cover $D$,
\begin{eqnarray*}
\mu(D)=\frac{1}{Z}\prod_{e\in D}W(e),
\end{eqnarray*}
where the product is over all edges present in $D$, and $Z$ is a
normalizing constant called the \textbf{partition function}, defined
to be
\begin{eqnarray*}
Z=\sum_{D}\prod_{e\in D}W(e),
\end{eqnarray*}
the sum over all dimer configurations of $G$.

For planar $\mathbb{Z}^2$-periodic graphs, endowed with periodic
weights on edges, the notion of Boltzman measure is replaced by that
of \textbf{Gibbs measure} with the property that if the dimer
configuration is fixed in an annular region, then the random dimer
configuration inside and outside the annulus are independent, and
the induced probability measure inside the annulus is the Boltzmann
measure defined above. It is known that different Gibbs measures may
be obtained as the infinite-volume, weak limits of Boltzmann
measures with various boundary conditions.

If we change the weight function $W$ by multiplying the edge weights
of all edges incident to a single vertex $v$ by the same constant,
the probability measure defined above does not change. So we define
two weight functions $W,W'$ to be \textbf{gauge equivalent} if one
can be obtained from the other by a sequence of such
multiplications.

The key objects used to obtain explicit expressions for the dimer
model are \textbf{Kasteleyn matrices}. They are weighted, oriented
adjacency matrices of the graph $G$ defined as follows. A
\textbf{clockwise odd orientation} of $G$ is an orientation of the
edges such that for each face (except the infinite face) an odd
number of edges pointing along it when traversed clockwise. For a
planar graph, such an orientation always exists \cite{ka2}. The
Kasteleyn matrix corresponding to such a graph is a
$|V(G)|\times|V(G)|$ skew-symmetric matrix $K$ defined by
\begin{align*}
K_{u,v}=\left\{\begin{array}{cc}W(uv)&{\rm if}\ u\sim v\ {\rm and}\
u\rightarrow v \\-W(uv)&{\rm if}\ u\sim v\ {\rm and}\ u\leftarrow
v\\0&{\rm else}.
\end{array}\right.
\end{align*}
It is known \cite{ka1,ka2,tes,kos} that for a planar graph with a
clock-wise odd orientation, the partition function of dimers
satisfies
\begin{align*}
Z=\sqrt{\det K}.
\end{align*}

Now let $G$ be a $\mathbb{Z}^2$-periodic planar graph. By this we
mean $G$ is embedded in the plane so that translations in
$\mathbb{Z}^2$ act by weight-preserving isomorphisms of
$G$, i.e. isomorphisms which map each edge to an edge with the same
weight. Let $G_n$ be the quotient graph $G/(n\mathbb{Z}\times
n\mathbb{Z})$. It is a finite graph on a torus. Let
$\gamma_{x,n}(\gamma_{y,n})$ be a path in the dual graph of $G_n$
winding once around the torus horizontally(vertically). Let
$E_H(E_V)$ be the set of edges crossed by $\gamma_{x,n}(\gamma_{y,n})$. We
give a \textbf{crossing orientation} for the toroidal graph $G_n$ as
follows. We orient all the edges of $G_n$ clockwise odd except for those in
$E_H\cup E_V$. This is possible since no other edges are crossing.
Then we orient the edges of $E_H$ clock-wise odd as if $E_V$ did not exist. Again
this is possible since $G-E_V$ is planar. To complete the
orientation, we also orient the edges of $E_V$ clockwise odd as if $E_H$ did not
exist.

For $\theta,\tau\in\{0,1\}$, let $K_n^{\theta,\tau}$ be the
Kasteleyn matrix $K_n$ in which the weights of edges in $E_H$ are
multiplied by $(-1)^{\theta}$, and those in $E_V$ are multiplied by
$(-1)^{\tau}$. It is proved in \cite{tes} that the partition
function $Z_n$ of the graph $G_n$ is
\begin{align*}
Z_n=\frac{1}{2}|\mathrm{Pf}(K_n^{00})+\mathrm{Pf}(K_n^{10})+\mathrm{Pf}(K_n^{01})-\mathrm{Pf}(K_n^{11})|.
\end{align*}

Let $E_m=\{e_1=u_1v_1,...,e_m=u_mv_m\}$ be a subset of edges of
$G_n$. Kenyon \cite{ke2} proved that the probability of these edges
occurring in a dimer configuration of $G_n$ with respect to the
Boltzmann measure $P_n$ is
\begin{align*}
P_n(e_1,...,e_m)=\frac{\prod_{i=1}^{m}W(u_iv_i)}{2Z_n}|\mathrm{Pf}(K_n^{00})_{E^{c}_{m}}+\mathrm{Pf}(K_n^{10})_{E^{c}_{m}}+\mathrm{Pf}(K_n^{01})_{E^{c}_{m}}-\mathrm{Pf}(K_n^{11})_{E^{c}_{m}}|,
\end{align*}
where $E_m^c=V(G_n)\setminus\{u_1,v_1,...,u_m,v_m\}$, and
$(K_n^{\theta\tau})_{E_m^c}$ is the submatrix of $K_n^{\theta\tau}$
whose lines and columns are indexed by $E_m^c$.

The asymptotic behavior of $Z_n$ when $n$ is large is an interesting
subject. One important concept is the partition function per
fundamental domain, which is defined to be
\[\lim_{n\rightarrow\infty}(Z_n)^{\frac{1}{n^2}}.\]
The logarithm of the partition function per fundamental domain is the \textbf{free energy}, namely
\begin{eqnarray*}
F:=\lim_{n\rightarrow\infty}\frac{1}{n^2}\log Z_n.
\end{eqnarray*}

Let $K_1$ be a Kasteleyn matrix for the graph $G_1$. Given any
parameters $z, w$, we construct a matrix $K(z,w)$ as follows. Let
$\gamma_{x,1}$, $\gamma_{y,1}$ be the paths introduced as above.
Multiply $K_{1,(u,v)}$, the $(u,v)$ entry of the matrix $K_1$, by $z$, if the crossing orientation on that edge is
from $u$ to $v$, otherwise multiply $K_{1,(u,v)}$ by $\frac{1}{z}$, and
similarly for $w$ on $\gamma_y$. Define the \textbf{characteristic
polynomial} $P(z,w)=\det K(z,w)$. The \textbf{spectral curve} is
defined to be the locus $\{(z,w)\in\mathbb{C}^{*}\times\mathbb{C}^{*}: P(z,w)=0\}$.

Gauge equivalent dimer weights give the same spectral curve. That is
because after Gauge transformation, the determinant multiplies by a
nonzero constant, thus not changing the locus of $P(z,w)$.

A formula for enlarging the fundamental domain is proved in
\cite{ckp,kos}. Let $P_n(z,w)$ be the characteristic polynomial of
$G_n$, and $P(z,w)$ be the characteristic polynomial of $G_1$,
then
\begin{eqnarray*}
P_n(z,w)=\prod_{u^n=z}\prod_{v^n=w}P(u,v).
\end{eqnarray*}

\subsection{Duality Transformation}

We call a geometric figure built with a certain number of bonds a closed polygon if at every lattice point, only an even number of bonds occurs. We associate to each square grid a 
Fisher graph by replacing each vertex with a gadet, as illustrated in the following 
figure

\begin{figure}[htbp]
\centering
\includegraphics*{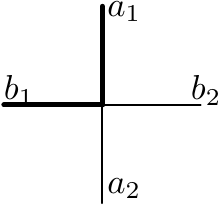}\qquad\qquad
\includegraphics*{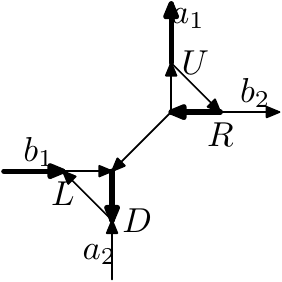} \caption{vertex
of Ising model and corresponding dimer gadget}
\end{figure}

If we draw a larger picture we will see that the Fisher graph is exactly the one obtained by replacing each vertex of the hexagonal lattice by a triangle. There is a one-to-one correspondence between closed polygon configurations on the square grid and dimer configurations on the Fisher graph. An edge is present in the closed polygon configuration of the square grid if and only if the corresponding edge is present in the dimer configuration of the Fisher graph, see Figure 2.

There are two ways to associate an Ising spin system on a square grid to the closed polygon configurations on the square grid. As a result, an Ising spin system on a square grid can be associated to the dimer system on a Fisher graph in two different ways. These two different ways lead to a duality transformation, which plays an important role in the identification of the critical temperature, as we will see. Note that these correspondences holds when there is no external magnetic field. Namely, the Ising model is defined by (\ref{probability}), given $h=0$, in order to introduce the correspondences between the Ising model and the dimer model.

\subsubsection{Correspondence 1: Measure-Preserving Correspondence }

An Ising model on a square grid is associated to closed polygon configurations on the \textbf{dual} square grid as follows: if two adjacent spins are of the same sign, the dual edge separating them is not present in the closed polygon configuration, otherwise the dual edge is present in the closed polygon configuration. An example of this correspondence is illustrated in Figure 3:

\begin{figure}[htbp]
  \centering
\includegraphics*{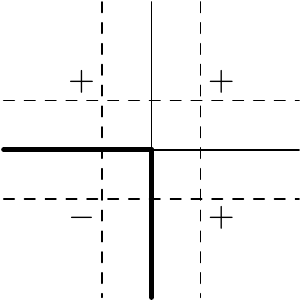}
   \caption{Measure-Preserving Correspondence}
\end{figure}

The dashed line are bonds of the Ising model, and the black line are edges of the dual grid. If two adjacent spins are of the same sign, the dual edge separating them is not present in the closed polygon configuration; otherwise the dual edge is present in the closed polygon configuration. This correspondence is two-to-one because negating the sign of all spins will end up with the same closed polygon configuration on the square grid. Combining with the one-to-one correspondence of closed polygon configurations on the square grid and the dimer configurations on Fisher graph, by assigning edge weights $e^{-2J_e}$ to the \textbf{dual} edges, we get a measure-preserving correspondence between the Ising spin system and the dimer system, namely,

\begin{eqnarray*}
Pr(D(\sigma))=Pr(\sigma)+Pr(-\sigma)=2Pr(\sigma),
\end{eqnarray*}
where $D(\sigma)$ is the dimer configuration corresponding to the spin configuration $\sigma$.

\subsubsection{Correspondence 2: Correspondence Based on High-Temperature Expansion}

Let $\mathcal{G}$ be a finite square grid, the high temperature expansion of the Ising model on $\mathcal{G}$ is
\begin{eqnarray*}
Z_{G,I}&=&\sum_{\sigma}\prod_{e=uv\in
E(G)}\exp(J_e\sigma_u\sigma_v)\\
&=&\sum_{\sigma}\prod_{e=uv\in E(G)}(\cosh
J_e+\sigma_u\sigma_v\sinh J_e)\\
&=&(\prod_{e=uv\in E(G)}\cosh J_e)\sum_{\sigma}\prod_{e=uv\in
E(G)}(1+\sigma_u\sigma_v\tanh J_e)\\
&=&(\prod_{e=uv\in E(G)}\cosh J_e)\sum_{C\in S}\prod_{e\in
C}2^{n^2}\tanh J_e,
\end{eqnarray*}
where $S$ is the set of all closed polygon configurations of $\mathcal{G}$. This way, up to a multiplicative constant, the Ising partition function is the same as the partition function of the closed polygon configuration of the $\mathbf{same}$ square grid. There is no one-to-one correspondence between configurations in this case, hence although in this case the partition function is invariant, it is not measure-preserving. Combining with the correspondence between the closed polygon configurations on the square grid and the dimer configurations on the Fisher graph, we obtain a correspondence between Ising models on square grid and dimer models on Fisher graph, by giving the corresponding edge of the Fisher graph $\tanh J_e$.

\subsection{Duality Transformation}

Let $G_n$ be the quotient graph of the square grid on the plane, as
defined on Page 4. Let $G_n^*$ be the dual graph of $G_n$. Define an
Ising model on $G_n$ with interactions $\{J_e\}_{e\in E(G_n)}$.
Assume the Ising model on $G_n$ has partition function $Z_{G_n, I}$.
Then $Z_{G_n,I}$ can be written as, up to a constant multiple, following from the measure-preserving correspondence, we have
\[Z_{G_n,I}=2\prod_{e\in E(G_n)}\exp(J_e)\sum_{C^*\in S_{00}^*}\prod_{e\in C^*}\exp(-2J_e):=2\prod_{e\in E(G_n)}\exp(J_e)Z_{F_n,D_{00}},\]
where $S_{00}^*$ is the set of closed polygon configurations of
$G_n^*$, with an even number of occupied bonds crossed by both
$\gamma_x$ and $\gamma_y$. The sum is over all configurations in
$S_{00}^*$.  Similarly, we can define $S_{01}^*(S_{10}^*,S_{11}^*)$
to be the set of closed polygon configurations of $G_n^*$, with an
even(odd,odd) number of occupied bonds crossed by $\gamma_x$, and an
odd(even,odd) number of occupied bonds crossed by $\gamma_y$. $F_n$ is the Fisher graph obtained from $G_n^*$ by replacing each vertex with a gadget, as described in Figure 2. Let
$Z_{F_n, D}$ be the partition function of dimer configurations on
$F_n$,  with
weights $e^{-2J_e}$ on dual edges of $G_n$, and weight 1 on all the
other edges. Then
\[Z_{F_n,D}=Z_{F_n,D_{00}}+Z_{F_n, D_{01}}+Z_{F_n,D_{10}}+Z_{F_n,D_{11}},\]
where
\[Z_{F_n,D_{\theta,\tau}}=\sum_{C^*\in S_{\theta,\tau}^*}\prod_{e\in C^{*}}\exp(-2J_e).\]
For example, $Z_{F_n,D_{01}}$ is the dimer partition function on
$F_n$ with an even number of occupied edges crossed by $\gamma_x$,
and an odd number of occupied edges crossed by $\gamma_y$. It also
corresponds to an Ising model which has the same configuration on
the two boundaries parallel to $\gamma_y$, and the opposite
configurations on the two boundaries parallel to $\gamma_x$. Similar
results hold for all the $Z_{F_n, D_{\theta,\tau}}$,
$\theta,\tau\in\{0,1\}$.

On the other hand, if we consider the high temperature expansion of
the Ising model on $G_n$, we have
\begin{eqnarray*}
Z_{G_n,I}=(\prod_{e=uv\in E(G_n)}\cosh J_e)\sum_{C\in S}\prod_{e\in
C}2^{n^2}\tanh J_e,
\end{eqnarray*}
where $S$ is the set of all closed polygon configurations of $G_n$.
Let $\tilde {F}_n$ be a Fisher graph embedded into an $n\times n$
torus, obtained from $G_n$ by the correspondence in Figure 3, with weights $\tanh J_e$ on edges of $G_n$, and weight 1 on
all the other edges. In other words, the edge with weight $\tanh
J_e$ of $\tilde{F}_n$ and the edge with weight $e^{-2J_e}$ of $F_n$
are dual edges. Then we have
\[Z_{G_n,I}=2\prod_{e\in E(G_n)}\exp(J_e)Z_{F_n,D_{00}}=2^{n^2}\prod_{e\in E(G_n)}\cosh J_e Z_{\tilde {F}_n,D},\]
where $Z_{\tilde{F}_n,D}$ is the partition function
of dimer configurations on $\tilde{F}_n$.

Similarly, we can expand all the $Z_{F_n,D_{\theta\tau}}$ as
follows:
\begin{equation}
Z_{F_n,D_{\theta,\tau}}=\frac{1}{2^{n^2+1}}\prod_{e\in
E_{G_n}}(1+\exp(-2J_e))Z_{\tilde{F}_n,D}((-1)^{\tau},(-1)^{\theta}).\label{duality}
\end{equation}
$Z_{\tilde{F}_{n,D}}(-1,1)$ is the dimer partition function of
$\tilde{F}_n$ with weights of edges dual to the edges crossed by $\gamma_x$ multiplied
by $-1$.   Similarly for $Z_{\tilde{F}_{n,D}}(1,-1)$ and
$Z_{\tilde{F}_{n,D}}(-1,-1)$. To see why (\ref{duality}) is true, let us consider, for example $Z_{F_n,D_{10}}$, an odd number of present edges are crossed by $\gamma_x$, and an even number of present edges are crossed by $\gamma_y$. The corresponding Ising model has boundary condition such that winding once along $\gamma_x$, the spins change sign, while winding once along $\gamma_y$, the spins are invariant. This is the same as for a row of edges crossed by $\gamma_y$, we change the coupling constant from $J_e$ to $-J_e$, and give the periodic boundary conditions. After the duality transformation, the edge weights $\tanh J_e$, if $e$ is crossed by $\gamma_y$ will be the opposite.

Without loss of generality, assume $n$ is even. Let $K(z,w)$ be the
Kasteleyn matrix of the toroidal graph, as defined on page 4. Given
the orientation in Figure 1, we have 
\footnotesize
\begin{eqnarray}
\mathrm{Pf}K_{F_n}(1,1)&=&Z_{F_n,D_{00}}-Z_{F_n,D_{01}}-Z_{F_n,D_{10}}-Z_{F_n,D_{11}}\nonumber\\ 
&=&\frac{1}{2^{n^2+1}}\prod_{e\in
E(G_n)}(1+\exp(-2J_e))[Z_{\tilde{F}_{n,D}}(1,1)-Z_{\tilde{F}_{n,D}}(-1,1)-Z_{\tilde{F}_{n,D}}(1,-1)-Z_{\tilde{F}_{n,D}}(-1,-1)] \nonumber\\
&=&\frac{1}{2^{n^2+2}}\prod_{e\in
E(G_n)}(1+\exp(-2J_e))\nonumber\\
&&\{[-\mathrm{Pf}K_{\tilde{F}_{n}}(1,1)+\mathrm{Pf}K_{\tilde{F}_{n}}(1,-1)+\mathrm{Pf}K_{\tilde{F}_{n}}(-1,1)+\mathrm{Pf}K_{\tilde{F}_{n}}(-1,-1)]\nonumber\\
&&-[-\mathrm{Pf}K_{\tilde{F}_{n}}(-1,1)+\mathrm{Pf}K_{\tilde{F}_{n}}(-1,-1)+\mathrm{Pf}K_{\tilde{F}_{n}}(1,1)+\mathrm{Pf}K_{\tilde{F}_{n}}(1,-1)]\nonumber\\
&&-[-\mathrm{Pf}K_{\tilde{F}_{n}}(1,-1)+\mathrm{Pf}K_{\tilde{F}_{n}}(1,1)+\mathrm{Pf}K_{\tilde{F}_{n}}(-1,-1)+\mathrm{Pf}K_{\tilde{F}_{n}}(-1,1)]\nonumber\\
&&-[-\mathrm{Pf}K_{\tilde{F}_{n}}(-1,-1)+\mathrm{Pf}K_{\tilde{F}_{n}}(-1,1)+\mathrm{Pf}K_{\tilde{F}_{n}}(1,-1)+\mathrm{Pf}K_{\tilde{F}_{n}}(1,1)]\}\nonumber\\
&=&-\frac{1}{2^{n^2}}\prod_{e\in
E(G_n)}(1+\exp(-2J_e))\mathrm{Pf}K_{\tilde{F}_{n}}(1,1).\label{dualitysignchange}
\end{eqnarray}
\normalsize
If we perform the same expansion for all the
$\mathrm{Pf}K_{F_n}((-1)^{\theta},(-1)^{\tau})$, for
$\theta,\tau\in\{0,1\}$, we get
\[\mathrm{Pf}K_{F_n}((-1)^{\theta},(-1)^{\tau})=\frac{1}{2^{n^2}}\prod_{e\in
E(G_n)}(1+\exp(-2J_e))\mathrm{Pf}K_{\tilde{F}_{n}}((-1)^{\theta},(-1)^{\tau}),\qquad
\mathrm{if}\ \{\theta,\tau\}\neq(0,0).\]

\begin{proposition}Consider a Fisher graph $F_n$ embedded into an $n\times
n$ torus. If an edge $e$ is parallel to $\gamma_x,\gamma_y$, give
$e$ weight $w_e=e^{-2J_e}$. For all the other edges give weight 1.
Define the duality transformation of $F_n$ to be another Fisher
graph $\tilde{F}_n$, embedded into an $n\times n$ torus. If an edge
$e$ is parallel to $\gamma_x$, $\gamma_y$, give $e$ weight $\tanh
J_e$. For all the other edges, give weight 1. Let $P(z,w)=0$ denote
the spectral curve. Then
\[P_{F_n}((-1)^{\theta},(-1)^{\tau})=0\Longleftrightarrow P_{\tilde{F}_n}((-1)^{\theta},(-1)^{\tau})=0.\]
\end{proposition}

\section{Critical Temperature and Spectral Curve}

The purpose of this section is to identify the critical temperature of the periodic Ising model with the condition that the spectral curve as a single real node on $\mathbb{T}^2$. Using the correspondence based on the high-temperature expansion, we express the spin-spin correlation of the periodic Ising model as the determinant of a block Toeplitz matrix. The analysis of the holomorphic properties of the matrix leads to the result.

\begin{lemma}(Kelly and Sherman\cite{gri}) Given a Hamiltanonian
\begin{eqnarray*}
\mathcal{H}=-\sum_{A\subset\Lambda}J_A\sigma_A,
\end{eqnarray*}
with ferromagnetic interactions, that is
\begin{eqnarray*}
J_A\geq 0,
\end{eqnarray*}
for all $A\subset\Lambda$, then
\begin{eqnarray*}
\frac{1}{\beta}\frac{\partial\langle\sigma_B\rangle}{\partial
J_C}=\langle\sigma_B\sigma_C\rangle-\langle\sigma_B\rangle\langle\sigma_C\rangle\geq
0.
\end{eqnarray*}
\end{lemma}

The high temperature expansion of the partition function gives
\begin{eqnarray*}
Z=\prod_e 2\cosh J_e\sum_{\mathcal{C}}\prod_{e\in\mathcal{C}}\tanh
J_e,
\end{eqnarray*}
where the sum is over all closed polygon of the square grid.
Assume $\tau_e=\tanh J_e$. We consider the spin located at $(0,0)$,
$\sigma_{00}$ and the spin located at $(0,N)$, $\sigma_{0N}$. Then
\begin{eqnarray*}
(\sigma_{00}\sigma_{0N})=(\sigma_{00}\sigma_{01})(\sigma_{01}\sigma_{02})...(\sigma_{0,N-1}\sigma_{0N}).
\end{eqnarray*}
Let $e_l$ denote the edge connecting $(0,l)$ and $(0,l+1)$. Since
\begin{eqnarray*}
\sigma_{0,l}\sigma_{0,l+1}(1+\tau_{e_l}\sigma_{0,l}\sigma_{0,l+1})=\tau_{e_l}(1+\frac{1}{\tau_{e_l}}\sigma_{0l}\sigma_{0,l+1}),
\end{eqnarray*}
we have
\begin{eqnarray*}
\langle\sigma_{00}\sigma_{0N}\rangle=Z^{-1}\prod_{e}2\cosh
J_e\prod_{l=0}^{N-1}\tau_{e_l}\sum_{\sigma}\prod_{l=0}^{N-1}(1+\frac{1}{\tau_{e_l}}\sigma_{0l}\sigma_{0,l+1}){\prod_{e}}'(1+\tau_e\sigma_p\sigma_q),
\end{eqnarray*}
where $\prod'$ means the terms corresponding to $e_l$, $0\leq l\leq
N-1$ are omitted. This expression is of the form of a partition
function for some Ising counting lattice with bonds $\tau_{e_l}$ on
the straight line connecting sites $(0,0)$ and $(0,N)$, replaced by
$\frac{1}{\tau_{e_l}}$. We transform the square grid $G_n$ to the Fisher Graph, using the technique as described in Figure 3. Assume the edge weights of the Fisher graph are $\tanh J_e$ for all edges corresponding to edges of $G_n$, and all the other edges have weight 1.

Let $\gamma_x(\gamma_y)$ denote a
path winding horizontally(vertically) once on the torus, multiply
all the edge weights crossed by $\gamma_x(\gamma_y)$ with $z$ or
$\frac{1}{z}$, ($w$ or $\frac{1}{w}$), according to the orientation
of edges.

Let $K_{mn}(z,w)$ denote the corresponding weighted adjacency
matrix. If both $m$ and $n$ are even, we have
\begin{eqnarray*}
\langle\sigma_{00}\sigma_{0N}\rangle_{m,n}=\prod_{l=0}^{N-1}\tau_{e_l}\frac{-\mathrm{Pf}
K'_{mn}(1,1)+\mathrm{Pf} K'_{mn}(-1,1)+\mathrm{Pf} K'_{mn}(1,-1)+\mathrm{Pf}
K'_{mn}(-1,-1)}{-\mathrm{Pf} K_{mn}(1,1)+\mathrm{Pf} K_{mn}(-1,1)+\mathrm{Pf} K_{mn}(1,-1)+\mathrm{Pf}
K_{mn}(-1,-1)},
\end{eqnarray*}
where $K'_{mn}(z,w)$ denote the corresponding weighted adjacency
matrix by changing the edge weights $\tau_{e_l}$ to
$\frac{1}{\tau_{e_l}}$.

\begin{theorem}For the Ising model on the infinite periodic square grid,
we can express
$\lim_{N\rightarrow\infty}\langle\sigma_{00}\sigma_{0N}\rangle^2$ as
the determinant of a block Toeplitz matrix multiplied by a function
of weights of all edges on the straight line connecting
$\sigma_{00}$ and $\sigma_{0N}$.
\end{theorem}
\begin{proof}
We consider $\frac{\mathrm{Pf} K'_{mn}(u,v)}{\mathrm{Pf} K_{mn}(u,v)}$ where
$u,v\in\{1,-1\}$. Assume $\delta=K'_{mn}(u,v)-K_{mn}(u,v)$, then
$\delta$ is given by
\begin{eqnarray*}
\delta(0,l;0,l+1)=-\delta^T(0,l+1;0,l)=\begin{array}{cc}&\begin{array}{cccc}R&L&U&D\end{array}\\
\begin{array}{c}R\\L\\U\\D\end{array}&\left(\begin{array}{cccc}0&\frac{1}{\tau_{e_l}}-\tau_{e_l}&0&0\\0&0&0&0\\0&0&0&0\\0&0&0&0\end{array}\right)\\\end{array}.
\end{eqnarray*}
If $0\leq l\leq N-1$ and zero otherwise. Here we are only interested in the entries labeled by $R,L,U,D$, because they are the only vertices adjacent to the edges with changed edge weights (from $\tau_{e_l}$ to $\frac{1}{\tau_{e_{l}}}$) when we compute the spin-spin correlation. Thus if we define $y$ as
the $2N\times 2N$ sub-matrix of $\delta$ in the subspace where
$\delta$ does not vanish identically, and if we define $Q$ as the
$2N\times 2N$ sub-matrix of $K_{mn}^{-1}$ in the same subspace, we find
\begin{eqnarray*}
\frac{\mathrm{Pf} K'_{mn}(u,v)}{\mathrm{Pf}
K_{mn}(u,v)}=(-1)^{k_1}\mathrm{Pf}(\delta+K_{mn}(u,v))\mathrm{Pf}(K_{mn}^{-1}(u,v))=(-1)^{k_2}\mathrm{Pf}y
\mathrm{Pf}(y^{-1}+Q(u,v)),
\end{eqnarray*}
where $k_1,k_2$ are constants depending only on $m,n,N$, and
independent of $u,v$, and $y$ is given by
\begin{eqnarray*}
y=\begin{array}{cc}&\begin{array}{cccccccc}0,0&0,1&...&0,N-1&0,1&0,2&...&0,N\\R&R&&R&L&L&&L
\end{array}\\ \begin{array}{cc}0,0&R\\0,1&R\\ \vdots&\\0,N-1&R\\0,1&L\\0,2&L\\
\vdots&\\0,N&L
\end{array}&\left(\begin{array}{cccc|cccc}&&&&\frac{1}{\tau_{e_1}}-\tau_{e_1}&&&\\&&&&&\frac{1}{\tau_{e_2}}-\tau_{e_2}&&\\ \multicolumn{4}{c|}{\raisebox{2ex}[0pt]{\Huge0}}&&\ddots&\\&&&&&&&\frac{1}{\tau_{e_{N}}}-\tau_{e_{N}}\\ \hline \tau_{e_1}-\frac{1}{\tau_{e_1}}&&&&&&&\\&{\tau_{e_2}}-\frac{1}{\tau_{e_2}}&&&&&&\\&&\dots&&\multicolumn{4}{|c}{\raisebox{2ex}[0pt]{\Huge0}}\\&&&\tau_{e_{N}}-\frac{1}{\tau_{e_{N}}}&&&&\end{array}\right),
\end{array}
\end{eqnarray*}
so that
\begin{eqnarray*}
\mathrm{Pf}y=\prod_{l=0}^{N-1}(\frac{1}{\tau_{e_l}}-\tau_{e_l})(-1)^{\frac{N(N-1)}{2}}.
\end{eqnarray*}
Hence we have
\begin{eqnarray*}
\langle\sigma_{00}\sigma_{0N}\rangle_{m,n}=\prod_{l=0}^{N-1}\tau_{e_l}\sum_{u,v\in\{-1,1\}}\frac{\mathrm{Pf}K'_{mn}(u,v)}{\mathrm{Pf}K_{mn}(u,v)}\frac{(-1)^{s(u,v)}\mathrm{Pf}K_{m,n}(u,v)}{Z},
\end{eqnarray*}
where $s(u,v)$ is $+1$ or $-1$ depending on $u,v$. We choose a
subsequence $m',n'\rightarrow\infty$ such that $k_2$ are always
even, then
\begin{eqnarray*}
\lim_{m',n'\rightarrow\infty}\langle\sigma_{00}\sigma_{0N}\rangle^2_{m,n}=\prod_{l=0}^{N-1}(1-\tau_{e_l}^2)^2
\det T_N.
\end{eqnarray*}
If we use $K$ to denote the weighted adjacency matrix of the
infinite graph, we can express $T_N$ as follows:
\begin{eqnarray*}\tiny
{T_N=\left(\begin{array}{cccccc}0&\cdots&K^{-1}_{0,0,R;0,N-1,R}&K^{-1}_{0,0,R;0,1,L}-\frac{\tau_{e_1}}{1-\tau_{e_1}^2}&\cdots&K^{-1}_{0,0,R;0,N,L}
\\K^{-1}_{1,1,R;0,0,R}&\cdots &K^{-1}_{0,1,R;0,N-1,R}&K^{-1}_{0,1,R;0,1,L}&\cdots&K^{-1}_{0,1,R;0,N,L}
\\ \vdots&\vdots&\vdots&\vdots&\vdots&\vdots
\\K^{-1}_{0,N-1,R;0,0,R}&\cdots&0&K^{-1}_{0,N-1,R;0,1,L}&\cdots& K^{-1}_{0,N-1,R;0,N,L}-\frac{\tau_{e_N}}{1-\tau_{e_N}^2}
\\K^{-1}_{0,1,L;0,0,R}+\frac{\tau_{e_1}}{1-\tau_{e_1}^2}&\cdots&K^{-1}_{0,1,L;0,N-1,R}&0&\cdots&K^{-1}_{0,1,L;0,N,L}
\\K^{-1}_{0,2,L;0,0,R}&\cdots&K^{-1}_{0,2,L;0,N-1,R}&K^{-1}_{0,2,L;0,1,L}&\cdots&K^{-1}_{0,2,L;0,N,L}\\
\vdots&\vdots&\vdots&\vdots&\vdots&\vdots\\
K^{-1}_{0,N,L;0,0,R}&\cdots
&K^{-1}_{0,N,L;0,N-1,R}+\frac{\tau_{e_N}}{1-\tau_{e_N}^2}&K^{-1}_{0,N,L;0,1,R}&\cdots&0
\end{array}\right)}.
\end{eqnarray*}
It is well-known that $T_N$ is independent of $u,v$, since the
entries of $K^{-1}$ can be expressed as follows
\begin{eqnarray}
K^{-1}(j,k,p;j',k',q)=\frac{1}{4\pi^2}\int_{0}^{2\pi}\int_{0}^{2\pi}e^{i\phi_1(k-k')}e^{i\phi_2(j-j')}K_{11}^{-1}(e^{i\phi_1},e^{i\phi_2})_{p,q}d\phi_1d\phi_2\label{inversekasteleyn},
\end{eqnarray}
where $K_{11}(z,w)$ is the weighted adjacency matrix of one
fundamental domain embedded into a torus. And
\begin{eqnarray}
\lim_{N\rightarrow\infty}\langle\sigma_{00}\sigma_{0N}\rangle^2=\lim_{N\rightarrow\infty}\det
T_N\prod_{0\leq l\leq N-1}(1-\tau_{e_l}^2)^2\label{spincorrelation}.
\end{eqnarray}
Let
$l_0$ denote the length of one period, assume $N$ is a multiple of
$l_0$.  We rearrange the rows and columns of $T_{N}$, in such a way that the determinant of $T_N$ will not change after this manipulation, neither will the limit as $N\rightarrow\infty$. Then we group the matrix after row and column rearrangements, denoted still by $T_N$, into $\frac{N}{l_0}\times \frac{N}{l_0}$ blocks, and each block is a $2l_0\times 2l_0$ submatrix, such that the $(j,j')$ block of $T_N$ has a row range from $(0,(j-1)l_0,R)$ to $(0,jl_0-1,R)$, from $(0,(j-1)l_0+1.L)$ to $(0,jl_0,L)$, and a column range from $(0,(j'-1)l_0,R)$ to $(0,j'l_0-1,R)$, from $(0,(j'-1)l_0+1.L)$ to $(0,j'l_0,L)$.

If we define $\psi$, a matrix-valued function on the unit
circle, as follows
\begin{eqnarray*}
\psi(\zeta)=\frac{1}{2\pi}\int_0^{2\pi}\left(\begin{array}{ccc}K^{-1}(\zeta,e^{i\phi_2})_{0,0,R;0,0,R}&\cdots&K^{-1}(\zeta,e^{i\phi_2})_{0,0,R;0,l_0-1,R}\\ \cdots&\cdots&\cdots \\K^{-1}(\zeta,e^{i\phi_2})_{0,l_0-1,R;0,0,R}&\cdots&K^{-1}(\zeta,e^{i\phi_2})_{0,l_0-1,R;0,l_0-1,R}\\K^{-1}(\zeta,e^{i\phi_2})_{0,1,L;0,0,R}+\frac{\tau_{e_1}}{1-\tau_{e_1}^2}&\cdots&K^{-1}(\zeta,e^{i\phi_2})_{0,1,L;0,l_0-1,R}\\ \cdots&\cdots&\cdots \\K^{-1}(\zeta,e^{i\phi_2})_{0,l_0,L;0,0,R}&\cdots&K^{-1}(\zeta,e^{i\phi_2})_{0,l_0,L;0,l_0-1,R}+\frac{\tau_{e_{l_0}}}{1-\tau_{e_{l_0}}^2}\end{array}\right.\\
\left.\begin{array}{ccc}K^{-1}(\zeta,e^{i\phi_2})_{0,0,R;0,1,L}-\frac{\tau_{e_{1}}}{1-\tau_{e_{1}}^2}&\cdots&K^{-1}(\zeta,e^{i\phi_2})_{0,0,R;0,l_0,L}\\
\cdots&\cdots&\cdots
\\K^{-1}(\zeta,e^{i\phi_2})_{0,l_0-1,R;0,1,L}&\cdots&K^{-1}(\zeta,e^{i\phi_2})_{0,l_0-1,R;0,l_0,L}-\frac{\tau_{e_{l_0}}}{1-\tau_{e_{l_0}}^2}\\K^{-1}(\zeta,e^{i\phi_2})_{0,1,L;0,1,L}&\cdots&K^{-1}(\zeta,e^{i\phi_2})_{0,1,L;0,l_0,L}\\
\cdots&\cdots&\cdots
\\K^{-1}(\zeta,e^{i\phi_2})_{0,l_0,L;0,1,L}&\cdots&K^{-1}(\zeta,e^{i\phi_2})_{0,l_0,L;0,l_0,L}\end{array}\right)
d\phi_2,
\end{eqnarray*}
then the $(j,j')$ block of $\psi$ is the $(j-j')$ Fourier coefficient of $\psi$, given (\ref{inversekasteleyn}). Let $T$ denote the infinite matrix as $N\rightarrow\infty$.  Namely,
\begin{eqnarray*}
(T_{N})_{j,j'}=\frac{1}{2\pi}\int_{|\zeta|=1}\zeta^{j-j'}\psi(\zeta)\frac{d\zeta}{i\zeta}.
\end{eqnarray*}
where the integral is evaluated entry by entry. Let $T$ denote the infinite matrix as $N\rightarrow\infty$, then
 $T$ is a
block Toeplitz matrix with symbol $\psi$.
\end{proof}

\begin{lemma}(\cite{li}) Let $\mathbb{T}^2=\{z,w||z|=1,|w|=1\}$ be the unit
torus, either $\det K(z,w)=0$ has no zeros on $\mathbb{T}^2$, or it
has a single real node on $\mathbb{T}^2$. That is, $\det K(z,w)=0$
intersects $\mathbb{T}^2$ at one of the points $(\pm 1,\pm 1)$, and
both $\frac{\partial\det K(z,w)}{\partial z}$ and
$\frac{\partial\det K(z,w)}{\partial w}$ vanish at that same point.
\end{lemma}

To study the asymptotic behavior of $\langle\sigma_{00}\sigma_{0N}\rangle$ as
$N\rightarrow\infty$, we first introduce the following lemmas
\begin{lemma}(Widom\cite{wi}) Let $T_n(\psi)$ be a block Toeplitz matrix with symbol $\psi$, and $n\times n$ blocks in total.  If
\begin{eqnarray*}
\sum_{k=-\infty}^{\infty}\|\psi_k\|+(\sum_{k=-\infty}^{\infty}|k|\|\psi_k\|^2)^{\frac{1}{2}}<\infty\\
\det\psi(e^{i\theta})\neq 0,\qquad \frac{1}{2\pi}\Delta_{0\leq\theta\leq
2\pi}\arg\det\psi(e^{i\theta})=0,
\end{eqnarray*}
here $\|\psi_k\|$ denotes the Hilbert-Schmidt norm of the matrix $\psi_k$, and  $\frac{1}{2\pi}\Delta_{0\leq\theta\leq
2\pi}\arg\det\psi(e^{i\theta})$ is the index of the point $0$ with respect to the curve $\{\det\psi(e^{i\theta}):0\leq \theta\leq 2\pi\}$, namely,
\begin{eqnarray*}
\frac{1}{2\pi}\Delta_{0\leq\theta\leq
2\pi}\arg\det\psi(e^{i\theta})=\frac{1}{2\pi i}\int_{|\zeta|=1}\frac{1}{\det\psi(\zeta)}\frac{\partial \det\psi(\zeta)}{\partial\zeta}{d\zeta},
\end{eqnarray*}
then
\begin{eqnarray*}
\lim_{n\rightarrow\infty}\frac{\det
T_n(\psi)}{G[\psi]^{n+1}}=E[\psi],
\end{eqnarray*}
with
\begin{eqnarray*}
G[\psi]&=&\exp\{\frac{1}{2\pi}\int_{0}^{2\pi}\log\det\psi(e^{i\theta})d\theta\},\\
E[\psi]&=&\det T[\psi]T[\psi^{-1}].
\end{eqnarray*}
where the last $\det$ refers to the determinant defined for
operators on Hilbert space differing from the identity by an
operator of trace class.
\end{lemma}

\begin{lemma}For ferromagnetic, periodic Ising model with nearest
neighbor pair interactions  on a square grid, if $\det K(z,w)$ has
no zeros on $\mathbb{T}^2$, then
\begin{eqnarray*}
\prod_{l=1}^{l_0}(1-\tau_{e_l}^2)^2\det\psi(e^{i\theta})=1,
\end{eqnarray*}
for any $\theta\in[0,2\pi)$.
\end{lemma}
\begin{proof} Let us consider an $m\times n$ torus, where $m$ and $n$
are the numbers of periods along each direction. The previous
technique to compute the spin-spin correlations implies that
\[\langle\sigma_{00}\sigma_{0n}\rangle=\prod_{1\leq l\leq l_0}\tau_{e_l}^{n}\frac{Z'_{mn}}{Z_{mn}},\]
where $Z'_{mn}$ is the partition function of dimer configurations on
an $m\times n$ Fisher graph with weights $\tau_{e_l}(1\leq l\leq n)$
changing to $\frac{1}{\tau_{e_l}}$. Since the graph is embedded into
an $m\times n$ torus, $\sigma_{00}$ and $\sigma_{0n}$ are actually
the same spin, therefore we have
$\langle\sigma_{00}\sigma_{0n}\rangle=1$, that is
\[\prod_{1\leq l\leq l_0}\tau_{e_l}^{n}Z'_{mn}=Z_{mn}.\]
Let $K'_{mn}(K_{mn})$ be the corresponding Kasteleyn matrix with
respect to $Z'_{mn}(Z_{mn})$. We are going to construct a correspondence between monomials in $\tau_{e_l}^2\det K'_{m,1}(1,w)$ and $\det K_{m,1}(-1,w)$, where $w$ is any number on the unit circle. Each term in the expansion of the determinant corresponds to an oriented loop configuration in the $m\times 1$ torus, where each vertex is adjacent to exactly two edges in the configuration. The correspondence is constructed by changing configurations on the first row of the graph, while keeping the configurations on all the other rows unchanged. Consider the first row. For each pair of nearest components, which are parts of non-doubled edge loops crossing the first row, there exists a pair of nearest triangles, one on each component. We change their paths through the pair of nearest triangles. This way, in the new configuration, all the $e_l$ occupied once originally will still be occupied once, with weight $\frac{1}{\tau_{e_l}}$ changed to $\tau_{e_l}$. Any $e_l$ unoccupied will be a doubled edge, with weight $1$ changed to $\tau_{e_l}^2$, any $e_l$ occupied twice will be unoccupied, with weight $\frac{1}{\tau_{e_l}^2}$ changed to 1. Hence the corresponding terms in $\prod_{1\leq l\leq l_0}\tau_{e_l}^2\det K_{m,1}'(1,w)$ and $\det K_{m,1}(-1,w)$ have the same absolute value. They are equal if the signs are also the same. For each component, which is part of a non-doubled-edge loop crossing the first row of the torus, if we change the path through a pair of triangles (corresponding to two nearest components on both sides), the sign will change by a factor of -1. This sign change will cancel if we change the sign of the z-edge passed by the loop simultaneously. Hence we have
\begin{eqnarray*}
\prod_{1\leq l\leq l_0}\tau_{e_l}^2\det K'_{m,1}(1,w)=\det K_{m,1}(-1,w).
\end{eqnarray*}
An example of such a path change is illustrated in Figure 4. Similarly we have
\begin{eqnarray*}
\prod_{1\leq l\leq l_0}\tau_{e_l}^2\det K_{m,1}'(-1,w)=\det K_{m,1}(1,w).
\end{eqnarray*}

\begin{figure}[htbp]
\centering
\scalebox{0.8}[0.8]{\includegraphics*{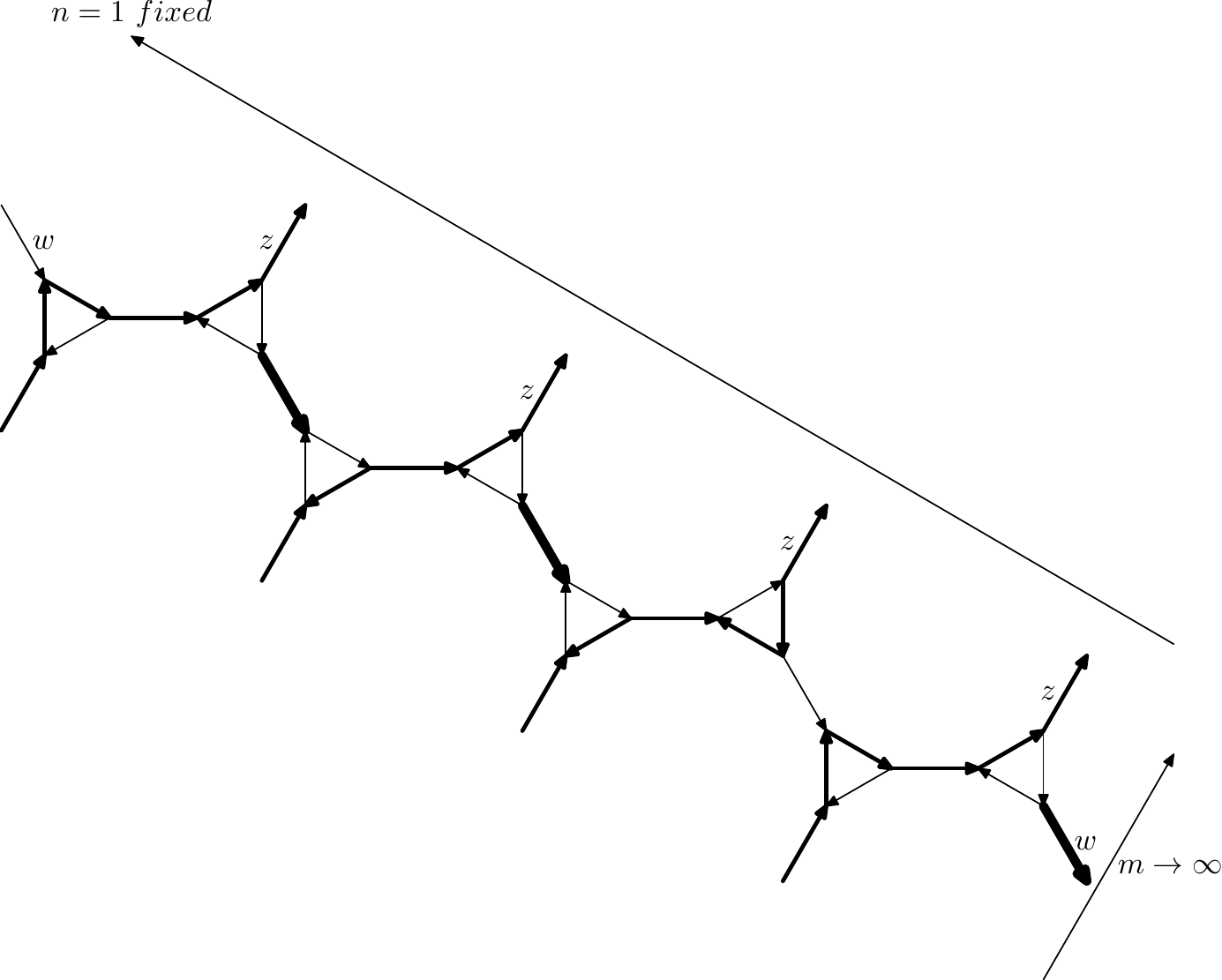}}\qquad\scalebox{0.8}[0.8]{\includegraphics*{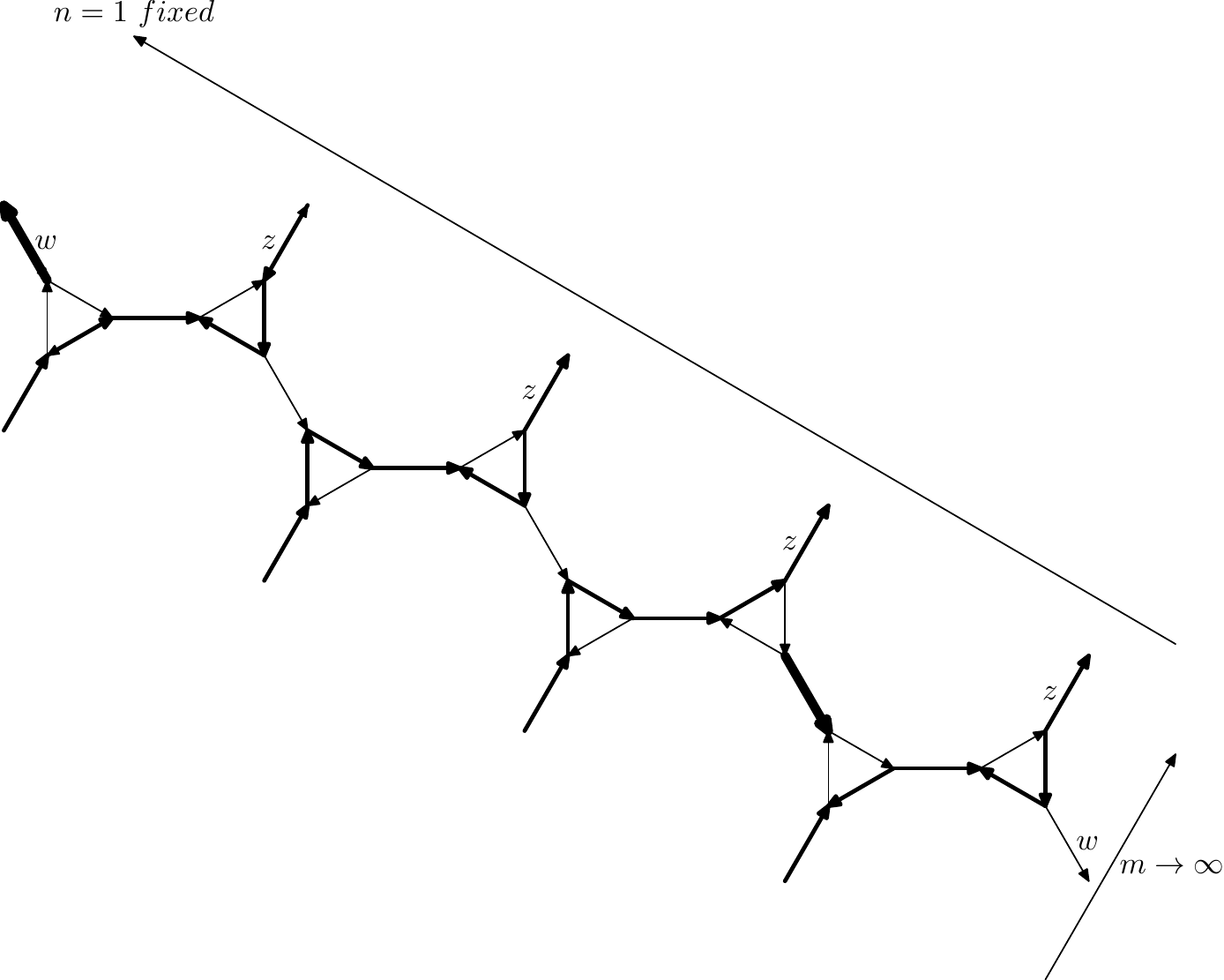}}
\caption{Change of configurations}
\end{figure}

Therefore, we have
\begin{eqnarray*}
[\prod_{l=1}^{l_0}(1-\tau_{e_l}^2)^2]\det
\psi(e^{i\theta})&=&\lim_{m\rightarrow\infty}\frac{(\prod_{l=1}^{l_0}\tau_{e_l}^2)\det
K'_{m,1}(1,e^{i\theta})}{\det K_{m,1}(1,e^{i\theta})}\\
&=&\lim_{m\rightarrow\infty}\frac{\det K_{m1}(-1,e^{i\theta})}{\det
K_{m1}(1,e^{i\theta})}.
\end{eqnarray*}

Let $v_1,\cdots,v_{2nl_0}$ be vertices which are endpoints of those edges changing weights from $\tau_{e_\ell}$ to $\frac{1}{\tau_{e_\ell}}$. Let us consider
\begin{eqnarray*}
\lim_{m\rightarrow\infty}\frac{\det K'_{m,1}(1,e^{i\theta})}{\det K_{m,1}(1,e^{i\theta})}&=&\lim_{m\rightarrow\infty}\sum_{u_1,\cdots,u_{2nl_0},u_j\sim v_j}\prod_{j=1}^{2n}K'_{m1}(1,e^{i\theta})_{v_ju_j}\frac{\det K'_{m1}(1,e^{i\theta})_{E^c\{u_1,\cdots,u_{2nl_0}\}}}{\det K_{m1}(1,e^{i\theta})}\\
&=&\lim_{m\rightarrow\infty}\sum_{u_1,\cdots,u_{2nl_0},u_j\sim v_j}\prod_{j=1}^{2n}K'_{m1}(1,e^{i\theta})_{v_ju_j}\det K^{-1}(1,e^{i\theta})_{E^c\{u_1,\cdots,u_{2nl_0}\}},
\end{eqnarray*}
where $K'_{m1}(1,e^{i\theta})_{E^c\{u_1,\cdots,u_{2nl_0}\}}$ is the submatrix of $K_{m1}'(1,e^{i\theta})$ deleting rows $v_1,\cdots,v_{2nl_0}$ and columns $u_1,\cdots,u_{2nl_0}$. This submatrix is also a submatrix of $K_{m1}(1,e^{i\theta})$, because all the edges with changed weights have been deleted. $K'_{m1}(1,e^{i\theta})_{v_j,u_j}$ are edge weights, which are independent of the first variable -  weights of edges crossed by $\gamma_x$. As $m\rightarrow\infty$, the limit of entries of $K_{m1}^{-1}(1,e^{i\theta})$ is independent of the first variable too, since $\lim_{m\rightarrow\infty}K^{-1}_{m1}(z,e^{i\theta})$ can always be expressed as the same complex integral no matter where $z$ is located on the unit circle. As a result,
\begin{eqnarray*}
\lim_{m\rightarrow\infty}\frac{\det K'_{m1}(1,e^{i\theta})}{\det K_{m1}(1,e^{i\theta})}=\lim_{m\rightarrow\infty}\frac{\det K'_{m1}(-1,e^{i\theta})}{\det K_{m1}(-1,e^{i\theta})}.
\end{eqnarray*}
Hence we have
\begin{eqnarray*}
\lim_{m\rightarrow\infty}\frac{\prod_{\ell=1}^{\ell_0}\tau_{e_\ell}^2\det K'_{m1}(1,e^{i\theta})}{\det K_{m1}(1,e^{i\theta})}&=&\lim_{m\rightarrow\infty}\frac{\det K_{m1}(-1,e^{i\theta})}{\det K_{m1}(1,e^{i\theta})}\\
&=&\lim_{m\rightarrow\infty}\frac{\det K_{m1}(1,e^{i\theta})}{\det K_{m1}(-1,e^{i\theta})}\\
&=&\lim_{m\rightarrow\infty}\frac{\prod_{\ell=1}^{\ell_0}\tau_{e_\ell}^2\det K'_{m1}(-1,e^{i\theta})}{\det K_{m1}(-1,e^{i\theta})}.
\end{eqnarray*}
Therefore
\begin{eqnarray}
[\prod_{\ell=1}^{\ell_0}(1-\tau_{e_\ell}^2)^2\det\psi(e^{i\theta})]^2=1\label{ga}.
\end{eqnarray}
When $\theta=0$, the quantity is the quotient of two skew-symmetric matrix, hence it is positive and equal to 1. The lemma follows from the continuity of $[\prod_{\ell=1}^{\ell_0}(1-\tau_{e_\ell}^2)^2\det\psi(e^{i\theta})]$.

\end{proof}

\begin{lemma}(Gohberg,Krein\cite{gk}) Let $\mathfrak{F}_1$ denote the set of
all completely continuous operator $A$ for which
\begin{eqnarray*}
\sum_{j=1}^{\infty}s_j(A)<\infty,
\end{eqnarray*}
where $s_j$ are singular values of $A$. Let $A(\mu)$ be an
operator-function with values in $\mathfrak{F}_1$ and holomorphic in
some region. Then the determinant $\det(I-A(\mu))$ is holomorphic in
the same region.
\end{lemma}

\begin{lemma} If we allow complex edge weights, then the spin
correlation
\[\lim_{N\rightarrow\infty}\langle\sigma_{00}\sigma_{0N}\rangle^2_p\] under toroidal boundary conditions is
analytic with respect to edge weights except for the edge weights
$\tanh J_{i,j}$ such that the spectral curve has a real node on
$\mathbb{T}^2$.
\end{lemma}
\begin{proof}
Let $H[\psi]$ denote the semi-infinite Hankel matrices with kernel
$\psi$, namely,
\begin{eqnarray*}
H[\psi]=(\psi_{i+j+1}),\ \ \ \ \ 0\leq i,j<\infty,
\end{eqnarray*}
then
\begin{eqnarray*}
T[\psi]T[\psi^{-1}]=I-H[\psi]H[\tilde{\psi}^{-1}],
\end{eqnarray*}
where
\begin{eqnarray*}
\tilde{\psi}(\zeta)=\psi(\frac{1}{\zeta}).
\end{eqnarray*}
Since $\psi$ is an even-order, anti-Hermitian matrix on the unit circle, its
determinant will always be real, hence we have
\begin{eqnarray*}
\Delta_{0\leq \theta\leq 2\pi}\arg\det\psi(e^{i\theta})=0.
\end{eqnarray*}
For any edge weights which do not lie on the critical surface, the
condition of Theorem~4.4, and 4.5 will satisfy(Widom\cite{wi}). Then we have
\begin{eqnarray*}
\lim_{N\rightarrow\infty}\langle\sigma_{00}\sigma_{0N}\rangle^2_p=\lim_{N\rightarrow\infty}\prod_{0\leq
l\leq N-1}(1-\tau_{e_l}^2)^2\det
T_N[\psi]=\det(I-H[\psi]H[\tilde{\psi}^{-1}]).
\end{eqnarray*}
Since the entries of $H[\psi]H[\tilde{\psi}^{-1}]$ will be analytic
everywhere except on the spectral curve, so is the spin correlation
$\lim_{N\rightarrow\infty}\langle\sigma_{00}\sigma_{0N}\rangle^2_p$.
\end{proof}

\begin{lemma}The weak limit of the Boltzmann measures
$\mathcal{P}_n$ on $G_n$, the quotient graph defined on page 4, is a
translation invariant Gibbs measure $\mathcal{P}$. The probability
of occurrence of a subset of edges $e_1=u_1v_1,...,e_m=u_mv_m$ of
$G$ in a dimer configuration of $G$ chose with respect to Gibbs
measure $\mathcal{P}$, is
\[\mathcal{P}(e_1,...,e_m)=\prod_{i=1}^{m}w_{u_iv_i}|\mathrm{Pf}(K^{-1})_{e_1,...,e_m}|,\]
where $K^{-1}$ is the inverse of the infinite Kasteleyn matrix,
defined as in (\ref{inversekasteleyn}), and $(K^{-1})_{e_1,...,e_m}$
is the sub-matrix of $K^{-1}$, whose lines and columns are defined
by the vertices of the edges $e_1,...,e_m$.
\end{lemma}
\begin{proof}The dimer partition function $Z_n$ of the graph $G_n$,
for $n$ even, is
\[Z_n=\frac{1}{2}[-\mathrm{Pf}(K_n^{00})+\mathrm{Pf}(K_n^{10})+\mathrm{Pf}(K_n^{01})+\mathrm{Pf}(K_n^{11})].\]

Let $E=\{e_1=u_1v_1,...,e_m=u_mv_m\}$ be a subset of edges of $G_n$,
then the probability $\mathcal{P}_n(e_1,...,e_m)$ of these edges
occurring in a dimer configuration of $G_n$ chosen with respect to
the Boltzmann measure $\mathcal{P}_n$, for $n$ even, is
\[\frac{\prod_{i=1}^{m}w_{u_iv_i}}{2Z_n}[-\mathrm{Pf}(K_n^{00})_{E^c}+\mathrm{Pf}(K_n^{10})_{E^c}+\mathrm{Pf}(K_n^{01})_{E^c}+\mathrm{Pf}(K_n^{11})_{E^c}],\]
where $E^c=V(G_n)\setminus\{u_1v_1,...,u_mv_m\}$ and
$(K_n^{\theta\tau})_{E^c}$ is the sub-matrix of $K_n^{\theta\tau}$
whose lines and columns are indexed by $E^c$. Let $P(z,w)=\det
K(z,w)$. By lemma 4.3, either the spectral curve $P(z,w)=0$ does not
intersect $\mathbb{T}^2$, or it has a single real node on
$\mathbb{T}^2$.

If $P(z,w)=0$ does not intersect $\mathbb{T}^2$,
\[(K_n^{\theta\tau})^{-1}_{(x,y,v),(x',y',v')}=\frac{1}{n^2}\sum_{j=0}^{n-1}\sum_{k=0}^{n-1}e^{\frac{i(2j+\theta)\pi(x-x')}{n}}e^{\frac{i(2k+\tau)\pi(y-y')}{n}}\frac{Cof[K(e^{\frac{i\pi(2j+\theta)}{n}},e^{\frac{i\pi(2k+\tau)}{n}})]_{v,v'}}{P(e^{\frac{i\pi(2j+\theta)}{n}},e^{\frac{i\pi(2k+\tau)}{n}})}.\]
The right side is a Riemann sum in non-critical case, it converges
to an integral for any $\theta,\tau\in\{-1,1\}$, that is,
\begin{align}
K^{-1}_{(x,y,v),(x',y',v')}&=&\lim_{n\rightarrow\infty}(K_n^{\theta,\tau})^{-1}_{(x,y,v),(x',y',v')}\\
&=&\frac{1}{4\pi^2}\iint_{\mathbb{T}^2}z^{x-x'}w^{y-y'}\frac{Cof[K(z,w)]_{v,v'}}{P(z,w)}\frac{dz}{iz}\frac{dw}{iw}\label{invkas}.
\end{align}
We also have
\[\mathrm{Pf}(K_n^{\theta\tau})^{-1}_{E}=(-1)^{\kappa(n)}\frac{(\mathrm{Pf}K_n^{\theta\tau})_{E^c}}{\mathrm{Pf}K_n^{\theta\tau}},\]
where $\kappa: \mathbb{N}\rightarrow \{0,1\}$ is a function. We can
choose a subsequence $\{n'\}$ such that $\kappa(n')$ is a constant
$\kappa_0$. Then
\begin{align*}
\mathcal{P}(e_1,...,e_m)&=&\lim_{n'\rightarrow\infty}\mathcal{P}_{n'}(e_1,...,e_m)\\
&=&\lim_{n'\rightarrow\infty}|(-1)^{\kappa_0}\frac{\prod_{i=1}^{m}w_{u_iv_i}}{2Z_n}[-\mathrm{Pf}(K_{n'}^{00})_E^{-1}\mathrm{Pf}K_{n'}^{00}+\mathrm{Pf}(K_{n'}^{01})_E^{-1}\mathrm{Pf}K_{n'}^{01}\\
&&+\mathrm{Pf}(K_{n'}^{10})_E^{-1}\mathrm{Pf}K_{n'}^{10}+\mathrm{Pf}(K_{n'}^{11})_E^{-1}\mathrm{Pf}K_{n'}^{11}]|\\
&=&\prod_{i=1}^{m}w_{u_iv_i}|\mathrm{Pf}(K^{-1})_{E}|.
\end{align*}
Note that the convergence actually holds for all $n$, because on the complement subsequence $\mathbb{N}\setminus\{n'\}$, $\kappa$ is also a constant given that $\kappa$ is a two-valued function. The sign does not matter since we take the absolute value in the end.

If $P(z,w)=0$ has a single real node on $\mathbb{T}^2$, see
\cite{bt} for a proof. The weak convergence of $\mathcal{P}_n$ for
$n\in \mathbb{N}$ follows from the fact that there is a unique
translation invariant Gibbs measure, and the weak limit of
$\mathcal{P}_n$ on any subsequence is translation invariant.
\end{proof}

\begin{lemma} Let $\mathcal{P}$ be the weak limit of $\mathcal{P}_n$,
Boltzmann measures defined for graph $G_n$ with periodic boundary
conditions.  $P(z,w)=0$ does not intersect $\mathbb{T}^2$, if and only if the
edge-edge correlation
\[\mathcal{P}(e_1\& e_2)-\mathcal{P}(e_1)\mathcal{P}(e_2)\]
converges to 0 exponentially fast as $|e_1-e_2|\rightarrow\infty$.
\end{lemma}
\begin{proof}By Lemma 4.8,
\begin{align}
\mathcal{P}(e_1\&
e_2)-\mathcal{P}(e_1)\mathcal{P}(e_2)&=&w_{u_1v_1}w_{u_2v_2}\left[\left|\mathrm{Pf}K^{-1}\left(\begin{array}{cccc}u_1&v_1&u_2&v_2\\u_1&v_1&u_2&v_2\end{array}\right)\right|-|K^{-1}_{u_1,v_1}K^{-1}_{u_2,v_2}|\right]\\
&=&w_{u_1v_1}w_{u_2v_2}[|K^{-1}_{u_1,v_1}K^{-1}_{u_2,v_2}+K^{-1}_{u_1,v_2}K^{-1}_{v_1,u_2}-K^{-1}_{u_1,u_2}K^{-1}_{v_1,v_2}|-|K^{-1}_{u_1,v_1}K^{-1}_{u_2,v_2}|].\label{covariance}
\end{align}
If $P(z,w)=0$ has no zeros on $\mathbb{T}^2$, (\ref{invkas}) and
(\ref{covariance}) imply that as $|e_1-e_2|\rightarrow\infty$,
$\mathcal{P}(e_1\&e_2)$ converges to 0 exponentially. If $P(z,w)$
has a real node on $\mathbb{T}^2$,  then
\[\frac{Cof(K(z,w))_{v,v'}}{P(z,w)}\]
is not an analytic function in edge weights, hence its Fourier coefficients does not decay exponentially fast.  To see that, if the Fourier coefficients decays exponentially, then the function is the limit of a locally uniform convergent power series, hence the function has to be analytic.
\end{proof}

\begin{lemma} If $P(z,w)=0$ has a single real node on $\mathbb{T}^2$, and if
\begin{eqnarray*}
\det\left.\left(\begin{array}{cc}\frac{\partial^2 P(e^{i\theta},e^{i\phi})}{\partial \theta^2}&\frac{\partial^2 P(e^{i\theta},e^{i\phi})}{\partial \theta\partial\phi}\\\frac{\partial^2 P(e^{i\theta},e^{i\phi})}{\partial \theta\partial\phi}&\frac{\partial^2 P(e^{i\theta},e^{i\phi})}{\partial \phi^2}\end{array}\right)\right|_{(\theta,\phi)=(0,0)}\neq 0,
\end{eqnarray*}
then the edge-edge correlation\begin{eqnarray*}
\mathcal{P}(e_1\&e_2)-\mathcal{P}(e_1)\mathcal{P}(e_2)
\end{eqnarray*}
converges to 0 quadratically as $|e_1-e_2|\rightarrow\infty$
\end{lemma}
\begin{proof}
Without loss of generality, assume the real node is $(1,1)$. Then
\begin{eqnarray*}
0=P(1,1)=(\mathrm{Pf}K(1,1))^2,
\end{eqnarray*}
where $\mathrm{Pf}K(1,1)$ is linear with respect to each edge weight. Fix an edge $e=uv$, and assume the weight of edge $e$ is $w_e$, then
\begin{eqnarray*}
P(1,1)=\mathcal{C}_2w_e^2+\mathcal{C}_1w_e+\mathcal{C}_0,
\end{eqnarray*}
where $\mathcal{C}_2$ is the signed weighted sum of loop configurations in which $e$ is a doubled edge. $\mathcal{C}_1$ consists of two parts: $\mathcal{C}_1=\mathcal{D}_1+\mathcal{D}_2$. $\mathcal{D}_1$ is the signed weighted sum of loop configurations in which $e$ is part of a loop oriented from $u$ to $v$; $\mathcal{D}_2$ is the signed weighted sum of loop configurations in which $e$ is part of a loop oriented from $v$ to $u$. When $z=w=1$, $\mathcal{D}_1=\mathcal{D}_2$. $\mathcal{C}_0$ is the signed weighted sum of loop configurations in which $e$ does not appear at all. The determinant of the submatrix of $K(1,1)$ by removing the row indexed by $u$ and the column indexed by $v$ is exactly $\mathcal{C}_2w_e+\mathcal{D}_1$.

Since $P(1,1)$ is a perfect square
\begin{eqnarray*}
(\mathcal{D}_1)^2=\mathcal{C}_0\mathcal{C}_2,
\end{eqnarray*}
as a result,
\begin{eqnarray*}
P(1,1)&=&(\mathcal{C}_2w_e+\mathcal{D}_1)w_e+(\mathcal{D}_1w_e+\frac{\mathcal{D}_1^2}{\mathcal{C}_2})\\
&=&\frac{1}{\mathcal{C}_2}(\mathcal{C}_2w_e+\mathcal{D}_1)^2.
\end{eqnarray*}
$P(1,1)=0$ implies $\mathcal{C}_2w_e+\mathcal{D}_1=0$, hence the determinant of a submatrix of $K(1,1)$ by removing the row indexed by $u$ and the column indexed by $v$ is 0, if $uv$ is an edge. Now, let us consider $Adj K(1,1)$, i.e. the adjugate matrix of $K(1,1)$ whose $(i,j)$ entry is $(-1)^{i+j}$ multiplying the determinant of the submatrix of $K(1,1)$ by removing the $j$th row and the $i$th column. Since $\det K(1,1)=0$, $Adj K(1,1)$ is a matrix of rank at most 1. As explained above, the entries of $Adj K(1,1)$ corresponding to an edge is 0. This implies that $Adj K(1,1)$ has 0 entries on each row and each column. Therefore $Adj K(1,1)$ is a zero matrix, and any minor of $K(1,1)$ by removing a row and a column is 0.

In a neighborhood of $(\theta,\phi)=(0,0)$ $P$ has the expansion
\begin{eqnarray*}
P(e^{i\theta},e^{i\phi})=a\theta^2+b\theta\phi+c\phi^2+\cdots,
\end{eqnarray*}
where $a,b,c\in\mathbb{R}$ are real numbers and $\cdots$ denotes terms of order at least 3. Then near a node, a point on $P=0$ satisfies either $\theta=\lambda\phi+O(\phi^2)$ or $\theta=\bar{\lambda}\phi+O(\phi^2)$, where $\lambda,\bar{\lambda}$, which are necessarily non-real, are the roots of $ax^2+bx+c=0$, since $P(z,w)>0$, for any $(z,w)\in\mathbb{T}^2\setminus\{(1,1)\}$(\cite{li}), and the assumption that the Hessian matrix at $(\theta,\phi)=(0,0)$ is invertible. This implies
\begin{eqnarray*}
b^2-4ac<0,\qquad
a=\left.\frac{\partial^2 P}{\partial\theta^2}\right|_{(\theta,\phi)=(0,0)}>0,\qquad
c=\left.\frac{\partial^2 P}{\partial\phi^2}\right|_{(\theta,\phi)=(0,0)}>0.
\end{eqnarray*}
Since $P(z,w)=\det K(z,w)$, we apply the differentiation rule for the determinant and obtain
\begin{eqnarray*}
\frac{\partial P}{\partial \theta}=\sum_{v}\det K_v,
\end{eqnarray*}
where $K_v$ is a matrix which are the same as $K$ except the row indexed by $v$. For each entry on the row indexed by $v$, we take the derivative with respect to $\theta$. The sum is over all the vertices of the quotient graph. Then we expand $\det K_v$ with respect to the row indexed by $v$. Obviously at most one entry on that row of $K_v$ is nonzero. Hence $\det K_v=0$ or $\det K_v=\pm i e^{\pm i\theta}(Adj K)_{v,w}$, where $vw$ is an $z$-edge or $\frac{1}{z}$-edge. Then we take the second order derivative at $\theta=0$, $\phi=0$, 
\begin{eqnarray*}
\frac{\partial\det K_v}{\partial\theta}=-e^{\pm i\theta}(Adj K(1,1))_{v,w}\pm e^{\pm i\theta}\left.\frac{d(Adj K)_{v,w}}{d\theta}\right|_{(\theta,\phi)=(0,0)}.
\end{eqnarray*}
The first term is 0. However, since $\left.\frac{\partial^2 P}{\partial\theta^2}\right|_{(\theta,\phi)=(0,0)}\neq 0$, there exists at least one $z$-edge $vv'$ such that $\left.\frac{d(Adj K)_{v,w}}{d\theta}\right|_{(\theta,\phi)=(0,0)}\neq 0$. 
 When $(\theta,\phi)$ lies in a neighborhood of $(0,0)$, we have
\begin{eqnarray*}
\frac{Cof[K(e^{i\theta},e^{i\phi})]_{v,v'}}{P(e^{i\theta},e^{i\phi})}&=&\frac{p\theta+q\phi}{a\theta^2+b\theta\phi+c\phi^2}+O(1)\\&=&\frac{1}{\lambda-\bar{\lambda}}\left[\frac{p\lambda-q}{a(\theta-\lambda\phi)}+\frac{q-p\bar{\lambda}}{a(\theta-\bar{\lambda}\phi)}\right]+O(1)\\&=&\frac{1}{\lambda-\bar{\lambda}}\left[\frac{p\lambda-q}{a(z-1)-a\lambda(w-1)}+\frac{q-p\bar{\lambda}}{a(z-1)-a\bar{\lambda}(w-1)}\right]+O(1),
\end{eqnarray*}
where $z=e^{i\theta},w=e^{i\phi}$. It is proved in $\cite{kos}$ that if
\begin{eqnarray*}
R(z,w)=\alpha(z-z_0)+\beta(w-w_0)+O(|z-z_0|^2+|w-w_0|^2),
\end{eqnarray*}
then we have the following asymptotic formula for the Fourier coefficients of $R^{-1}$
\begin{eqnarray*}
\frac{1}{(2\pi i)^2}\int_{\mathbb{T}^2}\frac{w^xz^y}{R(z,w)}\frac{dz}{z}\frac{dw}{w}=\frac{-w_0^xz_0^y}{2\pi i(x\alpha z_0-y\beta w_0)}+O(\frac{1}{x^2+y^2}).
\end{eqnarray*}
Therefore we have
\begin{eqnarray*}
&&\frac{1}{(2\pi i)^2}\int_{\mathbb{T}^2}\frac{w^xz^y Cof[K(z,w)]_{v,v'}}{P(z,w)}\frac{dz}{z}\frac{dw}{w}\\
&=&\frac{1}{\lambda-\bar{\lambda}}\left[\frac{q-p\lambda}{2\pi i(xa+ya\lambda)}+\frac{p\bar{\lambda}-q}{2\pi i(xa+ya\bar{\lambda})}\right]+O(\frac{1}{x^2+y^2}).
\end{eqnarray*}
Hence $K^{-1}_{v,v'}$ decays linearly as $|v-v'|\rightarrow\infty$, as a result, the edge-edge correlation decays quadratically as $|v-v'|\rightarrow\infty$.
\end{proof}

\begin{lemma}Let $\mathcal{P}_{I}$ be a Gibbs measure of the
ferromagnetic Ising model on the periodic square grid, $\mathcal{P}$
be the induced Gibbs measure of the dimer model by the Fisher
correspondence in Figure 2. If under $\mathcal{P}_{I}$ the spin-spin
correlation decays exponentially, then under $\mathcal{P}$, the
edge-edge correlation decays exponentially.
\end{lemma}
\begin{proof}
Let $\rho_i=\frac{\sigma_i+1}{2}$ be the lattice gas variable. Let
$A,B$ be two finite subsets of the vertices of the square grid.
Define
\begin{align*}
S_A=\sum_{i\in A}\rho_i,\\
\rho_A=\prod_{i\in A}\rho_i.
\end{align*}
Clearly $S_A$, $\rho_A$ and $S_A-\rho_A$ are nonnegative, monotone
functions of the $\{\rho_i\}$ and thus also of the $\{\sigma_i\}$.
According to the $F.K.G$ inequality, we have
\begin{align}
\langle\rho_B(S_A-\rho_A)\rangle\geq\langle\rho_B\rangle\langle
S_A-\rho_A\rangle\geq 0\label{inequality1},
\end{align}
and
\begin{align}
\langle S_A(S_B-\rho_B)\rangle\geq\langle S_A\rangle\langle
S_B-\rho_B\rangle\geq 0.\label{inequality2}
\end{align}
Combining (\ref{inequality1}) and (\ref{inequality2}), we have
\begin{align*}
0&\leq&
\langle\rho_A\rho_B\rangle-\langle\rho_A\rangle\langle\rho_B\rangle\leq\langle S_A\rho_B\rangle-\langle S_A\rangle\langle\rho_B\rangle\leq\langle S_AS_B\rangle-\langle S_A\rangle \langle S_B\rangle\\
&=&\sum_{i\in A}\sum_{j\in
B}\langle\rho_i\rho_j\rangle-\langle\rho_i\rangle\langle\rho_j\rangle=\frac{1}{4}\sum_{i\in
A}\sum_{j\in
B}\langle\sigma_i\sigma_j\rangle-\langle\sigma_i\rangle\langle\sigma_j\rangle.
\end{align*}
Let $e_1(e_2)$ be an edge with endpoints $u_1,v_1(u_2,v_2)$. According to the Fisher correspondence illustrated in Figure 2, we
have
\begin{align*}
\mathcal{P}(e_1\&e_2)-\mathcal{P}(e_1)\mathcal{P}(e_2)&=&\langle\rho_{u_1}(1-\rho_{v_1})\rho_{u_2}(1-\rho_{v_2})\rangle-\langle\rho_{u_1}(1-\rho_{v_1})\rangle\langle\rho_{u_2}(1-\rho_{v_2})\rangle\\
&&+\langle\rho_{v_1}(1-\rho_{u_1})\rho_{v_2}(1-\rho_{u_2})\rangle-\langle\rho_{v_1}(1-\rho_{u_1})\rangle\langle\rho_{v_2}(1-\rho_{u_2})\rangle\\
&&+\langle\rho_{u_1}(1-\rho_{v_1})\rho_{v_2}(1-\rho_{u_2})\rangle-\langle\rho_{u_1}(1-\rho_{v_1})\rangle\langle\rho_{v_2}(1-\rho_{u_2})\rangle\\
&&+\langle\rho_{v_1}(1-\rho_{u_1})\rho_{u_2}(1-\rho_{v_2})\rangle-\langle\rho_{v_1}(1-\rho_{u_1})\rangle\langle\rho_{u_2}(1-\rho_{v_2})\rangle,
\end{align*}
and
\begin{align*}
&&|\langle\rho_{u_1}(1-\rho_{v_1})\rho_{u_2}(1-\rho_{v_2})\rangle-\langle\rho_{u_1}(1-\rho_{v_1})\rangle\langle\rho_{u_2}(1-\rho_{v_2})\rangle|\\
&=&|(\langle\rho_{u_1}\rho_{u_2}\rangle-\langle\rho_{u_1}\rangle\langle\rho_{u_2}\rangle)-(\langle\rho_{u_1}\rho_{u_2}\rho_{v_2}\rangle-\langle\rho_{u_1}\rangle\langle\rho_{u_2}\rho_{v_2}\rangle)\\
&&-(\langle\rho_{u_1}\rho_{v_1}\rho_{u_2}\rangle-\langle\rho_{u_1}\rho_{v_1}\rangle\langle\rho_{u_2}\rangle)+(\langle\rho_{u_1}\rho_{v_1}\rho_{u_2}\rho_{v_2}\rangle-\langle\rho_{u_1}\rho_{v_1}\rangle\langle\rho_{u_2}\rho_{v_2}\rangle)|\\
&\leq&4|\langle\rho_{u_1}\rho_{u_2}\rangle-\langle\rho_{u_1}\rangle\langle\rho_{u_2}\rangle|+2|\langle\rho_{u_1}\rho_{v_2}\rangle-\langle\rho_{u_1}\rangle\langle\rho_{v_2}\rangle|\\
&&+2|\langle\rho_{u_2}\rho_{v_1}\rangle-\langle\rho_{u_2}\rangle\langle\rho_{v_1}\rangle|+|\langle\rho_{v_1}\rho_{v_2}\rangle-\langle\rho_{v_1}\rangle\langle\rho_{v_2}\rangle|.
\end{align*}
The other terms of $\mathcal{P}(e_1\&
e_2)-\mathcal{P}(e_1)\mathcal{P}(e_2)$ can be bounded in a similar
way. Hence if the spin-spin correlation decays exponentially, the
edge-edge correlation decays exponentially after the Fisher
correspondence described in Figure 2.
\end{proof}

\begin{lemma}Given interactions $\{J_k\}_{k=1}^{2n^2}$, when $\beta$
increases from $0$ to $+\infty$, $\beta_c$ corresponds to the
smallest $\beta$ that the spectral curve for dimer configurations
with weights $\tanh \beta J_k$ has a real node on $\mathbb{T}^2$.
\end{lemma}
\begin{proof}
 Consider the weak limit of
Boltzmann measures of the Ising model under periodic boundary
conditions,
\begin{eqnarray*}
\langle\sigma_i\sigma_j\rangle_p&=&\langle\sigma_i\sigma_j\rangle_p-\langle\sigma_i\rangle_p\langle\sigma_j\rangle_p\\
&=&4(\langle\rho_i\rho_j\rangle_p-\langle\rho_i\rangle_p\langle\rho_j\rangle_p)=4\langle\rho_i\rho_j\rangle_p-1.
\end{eqnarray*}

Since $\rho_i\rho_j$ is an increasing function with respect to each
single variable $\sigma$, $\langle\sigma_i\sigma_j\rangle_{p}$ is an
increasing function with respect to $\beta$. In particular
$\lim_{n\rightarrow\infty}\langle\sigma_{00}\sigma_{0n}\rangle_p$ is
increasing in $\beta$. There is a unique critical temperature
$\beta_{c,p}$, for any given interactions $\{J_{i,j}\}_{i,j}$, such
that
\begin{eqnarray*}
\lim_{n\rightarrow\infty}\langle\sigma_{00}\sigma_{0n}\rangle_p=0\qquad
\mathrm{if}\ \beta<\beta_{c,p},\\
\lim_{n\rightarrow\infty}\langle\sigma_{00}\sigma_{0n}\rangle_p>0\qquad
\mathrm{if}\ \beta>\beta_{c,p}.
\end{eqnarray*}
Hence at the edge weights $\tanh \beta_{c,p}J$,
$\lim_{n\rightarrow\infty}\langle\sigma_{00}\sigma_{0n}\rangle_{p}$
fails to be an analytic function of $\beta$(complex value of $\beta$
are also allowed). By Lemma 4.7, $\tanh \beta_{c,p}J$ is a subset of
the edge weights such that $P(z,w)=0$ has a real node on
$\mathbb{T}^2$. By Theorem 1.1, the uniqueness of the translation
invariant Gibbs measure implies that
\[\lim_{n\rightarrow\infty}\langle\sigma_{00}\sigma_{0n}\rangle_{+}=\lim_{n\rightarrow\infty}\langle\sigma_{00}\sigma_{0n}\rangle_p.\]
Therefore for $\beta<\beta_{c,p}$
\[0\leq \lim_{n\rightarrow\infty}\langle\sigma_{00}\sigma_{0n}\rangle_{+}-\langle\sigma_{00}\rangle_{+}^2\leq\lim_{n\rightarrow\infty}\langle\sigma_{00}\sigma_{0n}\rangle=0,\]
where the first inequality follows from Lemma 4.1. Hence for
$\beta<\beta_{c,p}$, the spontaneous magnetization
\[m^*=\frac{1}{|\Lambda_1|}\sum_{i\in\Lambda_1}\langle\sigma_i\rangle_{+}=0.\]
$\Lambda_1$, the quotient graph consisting of one period. Therefore
$\beta_{c,p}\leq \beta_c$. $\beta_c$, as defined before, is the
reciprocal of the lowest temperature such that $m^*=0$.

It is proved in \cite{abf,ai2} that for $\beta<\beta_c$, the
spin-spin correlation decays exponentially, therefore after the
Fisher correspondence, the edge-edge correlation on a graph with
weights $e^{-2\beta J}$ decays exponentially. We derive that
$\beta_c\leq \beta_{c,p}$, because at $\beta_{c,p}$, $P(z,w)=0$ has
a real node on $\mathbb{T}^2$, the edge-edge correlation does not decay exponentially, see Lemma 4.9. Combining with the previous argument, we have
$\beta_{c,p}=\beta_c$. When $\beta<\beta_c$, the corresponding
spectral curve with weights $e^{-2\beta J}$ does not intersect
$\mathbb{T}^2$. By proposition 3.1, the spectral curve of dimer
configurations with weights $\tanh \beta J$ does not intersect
$\mathbb{T}^2$ either, if $\beta<\beta_c$. This completes the proof.
\end{proof}

\begin{lemma}Assume the Ising model has period $n\times n$, $n$ is
even. If the weights of all edges parallel to $\gamma_x,\gamma_y$
lie in the open interval $(0,1)$, and all the other edges have
weight 1. Given an orientation to the Fisher graph as illustrated in
Figure 1, the only possible real node of the spectral curve
$P(z,w)=0$ on $\mathbb{T}^2$ is $(1,1)$.
\end{lemma}
\begin{proof}It suffices to prove that under the assumption of the
lemma, $\mathrm{Pf}K(1,-1)\neq 0$, $\mathrm{Pf}K(-1,1)\neq 0$ and
$\mathrm{Pf}K(-1,-1)\neq 0$. Given an orientation as in Figure 1,
\begin{equation}
\mathrm{Pf}K(-1,-1)=Z_{D,00}+Z_{D,01}+Z_{D,10}-Z_{D,11},\label{pf11}
\end{equation}
where $Z_{D,01}$ is the partition function of dimer configurations
with an even number of occupied edges crossed by $\gamma_x$ and an
odd number of occupied edges crossed by $\gamma_y$.
$Z_{D,00},Z_{D,10},Z_{D,11}$ are defined similarly. The expression
of $P(-1,-1)$ follows from the fact that if we reverse the
orientation of all the edges crossed by $\gamma_x$ and $\gamma_y$,
the resulting orientation is an crossing orientation, as defined on
Page 4. Similarly, we have
\begin{eqnarray}
\mathrm{Pf}K(1,-1)&=&Z_{D,00}+Z_{D,01}-Z_{D,10}+Z_{D,11}\label{pf01},\\
\mathrm{Pf}K(-1,1)&=&Z_{D,00}-Z_{D,01}+Z_{D,10}+Z_{D,11}\label{pf10}.
\end{eqnarray}
$Z_{D,01}$ also corresponds to an Ising partition function with the
opposite spin configurations along the two boundaries parallel to
$\gamma_y$, and the same spin configuration along the two boundaries
parallel to $\gamma_x$. After the duality transformation, $Z_{D,01}$
becomes, up to a constant multiple, the partition function of dimer
configurations on the dual Fisher graph $\tilde{F}$, with weights of
edges crossed by $\gamma_x$ multiplied by $-1$, see (\ref{duality}). Since the duality transformation does not change the partition function, up to a multiplicative constant, and after the duality transformation $Z_{D,00}$, $Z_{D,01}$, $Z_{D,10}$, $Z_{D,11}$ are partition function of the dimer model on the same graph with the same edge weights except on the boundary edges, and $Z_{D,00}$ is the only one of them such that all the edge weights are positive,
as a result, $Z_{D,00}$ is the biggest of
$\{Z_{D,\theta\tau}\}_{\theta,\tau\in\{0,1\}}$, because all the
others are partition function of dimer configurations of the same
graph with negative weights. As a result of (\ref{pf11}),
(\ref{pf01}) and (\ref{pf10}), and the fact that each one of $\{Z_{D,\theta\tau}\}_{\theta,\tau\in\{0,1\}}$ is positive as the partition function of dimer model on a graph with only positive edge weights,
\[\mathrm{Pf}K((-1)^{\theta},(-1)^{\tau})>0,\ if\ (\theta,\tau)\neq (0,0).\]
\end{proof}

\begin{lemma}Assume $n$ is even. For any given
$\{J_k\}_{k=1}^{2n^2}$, $0<J_k<\infty$, the curve
\[\gamma(t)=(\tanh tJ_1,\cdots,\tanh tJ_{2n^2})\qquad 0\leq t\leq+\infty\]
in the edge weight space(a $2n^2$-dimensional vector space with
coordinates given by edge weights) intersects $P(1,1)=0$ at a unique
point $\gamma(t_0)$, $0<t_0<+\infty$.
\end{lemma}
\begin{proof}
First of all, we claim that
\[\gamma(t)\bigcap \{P(1,1)=0\}\neq \emptyset.\]
Note that $P(1,1)=(\mathrm{Pf}K(1,1))^2$. Given edge weights
$\gamma(0)=(0,\cdots,0)$, $\mathrm{Pf}K(1,1)=1$. Consider the edge
weights $\lim_{t\rightarrow+\infty}\gamma(t)=(1,\cdots,1)$. After
the duality transformation, all the edges parallel to $\gamma_x$,
$\gamma_y$ have weight 0. Since $Z_{D,01}$, $Z_{D,10}$, $Z_{D,11}$
corresponds to dimer configurations by negating the weights
$z-$($w-$,$z-$ and $w-$) edges, and all such edges have weight 0, we
obtain
\[Z_{D,00}=Z_{D,01}=Z_{D,10}=Z_{D,11}=2^{n^2-1},\]
for weights $(1,\cdots,1)$, see (\ref{duality}). Hence
$\lim_{t\rightarrow+\infty}\mathrm{Pf}K(1,1)<0$. There exists
$0<t_0<\infty$, such that $\gamma(t_0)\in\{P(1,1)=0\}$.

Now we prove the uniqueness of the intersection. Since $P(1,1)=0$ is
invariant under the duality transformation, i.e, any edge weights
lie on the surface $P(1,1)=0$ if and only if the images after
duality transformation lie on the surface $P(1,1)=0$, it suffices to
prove that the image of $\gamma(t)$
\[\gamma^*(t)=(\exp(-2tJ_1),\cdots,\exp(-2tJ_{2n^2}))\]
intersects $P(1,1)=0$ at a unique point. Assume $f(t)$ is the
function obtained by plugging edge weights
$(\exp(-2tJ_1),\cdots,\exp(-2tJ_{2n^2}))$ into $\mathrm{Pf}K(1,1)$, and
$g(t)=2\prod_{k=1}^{2n^2}\exp(tJ_k)f(t)$, it suffices to prove that
there is a unique $t_0$ in $(0,+\infty)$ such that $g(t_0)=0$. Given
interactions $(tJ_1,\cdots,tJ_{2n^2})$ on edges of the quotient
graph of $G_n$, let $Z_{G_n,I}^{00}$ be the partition function of
the Ising model on $G_n$. Let
$Z_{G_n,I}^{10}(Z_{G_n,I}^{01},Z_{G_{n,I}}^{10},Z_{G_{n,I}}^{11})$
be the partition function of an Ising model, with the interactions
on edges crossed by $\gamma_x$($\gamma_y$, $\gamma_x$ and
$\gamma_y$) multiplied by $-1$. Then we have
\begin{eqnarray*}
g(t)&=&Z_{G_n,I}^{00}-Z_{G_n,I}^{01}-Z_{G_n,I}^{10}-Z_{G_n,I}^{11},\\
g'(t)&=&Z_{G_n,I}^{00}\sum_{e=uv\in
E(G_n)}J_e^{00}\langle\sigma_u\sigma_v\rangle_{00}-Z_{G_n,I}^{01}\sum_{e=uv\in
E(G_n)}J_e^{01}\langle\sigma_u\sigma_v\rangle_{01}\\&&-Z_{G_n,I}^{10}\sum_{e=uv\in
E(G_n)}J_e^{10}\langle\sigma_u\sigma_v\rangle_{10}-Z_{G_n,I}^{11}\sum_{e=uv\in
E(G_n)}J_e^{11}\langle\sigma_u\sigma_v\rangle_{11},
\end{eqnarray*}
where $\{J_e^{\theta,\tau}\}_{e\in E(G_n)}$ are corresponding
interactions for $Z_{G_n,I}^{\theta\tau}$ and
$\langle\quad\rangle_{\theta\tau}$ are the expected values with
respect to Boltzmann measure defined by $\{J_e^{\theta\tau}\}_{e\in
E(G_n)}$. Let $E_0$ be the subset of $E(G_n)$ consisting of edges
intersecting neither $\gamma_x$ nor $\gamma_y$, and $E_x(E_y)$ be
subsets of $E(G_n)$ consisting of edges crossed by
$\gamma_x(\gamma_y)$.
\begin{eqnarray*}
g'(t)&=&\sum_{e=uv\in
E_0}J_e^{00}(Z_{G_n,I}^{00}\langle\sigma_u\sigma_v\rangle_{00}-Z_{G_n,I}^{01}\langle\sigma_u\sigma_v\rangle_{01}-Z_{G_n,I}^{10}\langle\sigma_u\sigma_v\rangle_{10}-Z_{G_n,I}^{11}\langle\sigma_u\sigma_v\rangle_{11})\\
&&+\sum_{e=uv\in
E_x}J_e^{00}(Z_{G_n,I}^{00}\langle\sigma_u\sigma_v\rangle_{00}-Z_{G_n,I}^{01}\langle\sigma_u\sigma_v\rangle_{01}+Z_{G_n,I}^{10}\langle\sigma_u\sigma_v\rangle_{10}+Z_{G_n,I}^{11}\langle\sigma_u\sigma_v\rangle_{11})\\
&&+\sum_{e=uv\in
E_y}J_e^{00}(Z_{G_n,I}^{00}\langle\sigma_u\sigma_v\rangle_{00}+Z_{G_n,I}^{01}\langle\sigma_u\sigma_v\rangle_{01}-Z_{G_n,I}^{10}\langle\sigma_u\sigma_v\rangle_{10}+Z_{G_n,I}^{11}\langle\sigma_u\sigma_v\rangle_{11}).
\end{eqnarray*}
We claim that for $(\theta,\tau)\neq (0,0),$
\begin{equation}
|\langle\sigma_u\sigma_v\rangle_{\theta\tau}|<\langle\sigma_u\sigma_v\rangle_{00}\label{parityinequality}.
\end{equation}
First of all,
\begin{equation}
\langle\sigma_u\sigma_v\rangle_{\theta\tau}<\langle\sigma_u\sigma_v\rangle_{00},\label{parityinequality1}
\end{equation}
by F.K.G inequality. The inequality (\ref{parityinequality1}) is
strict because the graph is finite. Moreover,
\begin{eqnarray*}
\langle\sigma_u\sigma_v\rangle_{\theta\tau}+\langle\sigma_u\sigma_v\rangle_{00}&=&\frac{1}{Z_{G_n,I}^{00}Z_{G_n,I}^{\theta\tau}}\sum_{\sigma,\sigma'}[(\sigma_u
\sigma_v+\sigma'_u\sigma'_v)\exp(\sum_{e\in
E(G_n)}(J_e^{\theta\tau}+J_e^{00}))\sigma_u\sigma_v]\\
&=&\frac{1}{Z_{G_n,I}^{00}Z_{G_n,I}^{\theta\tau}}\sum_{t}(1+t_ut_v)\sum_{\sigma}[\sigma_u
\sigma_v\exp(\sum_{e\in
E(G_n)}(J_e^{\theta\tau}+J_e^{00}))\sigma_u\sigma_v],
\end{eqnarray*}
where we have introduced the Ising variables $\sigma'_i=\pm1$ an
$t_i=\sigma'_i\sigma_i$. Since $J_e^{00}+J_e^{\theta\tau}\geq 0$,
the ferromagnetic condition and finiteness of $G_n$ imply
\begin{equation}
\langle\sigma_u\sigma_v\rangle_{\theta\tau}+\langle\sigma_u\sigma_v\rangle_{00}>0.\label{parityinequality2}
\end{equation}
(\ref{parityinequality}) follows from (\ref{parityinequality1}) and
(\ref{parityinequality2}). Hence
\[g'(t)>(Z_{G_n,I}^{00}-Z_{G_n,I}^{01}-Z_{G_n,I}^{10}-Z_{G_n,I}^{11})\sum_{e=uv\in E(G_n)}J_e^{00}\langle\sigma_u\sigma_v\rangle_{00}.\]
Since $J_e^{00}>0$, $\langle\sigma_u\sigma_v\rangle_{00}>0$.
\begin{eqnarray*}
Z_{G_n,I}^{00}-Z_{G_n,I}^{01}-Z_{G_n,I}^{10}-Z_{G_n,I}^{11}=\frac{\mathrm{Pf}K(1,1)}{2\prod_{k=1}^{2n^2}\exp(tJ_k)}.
\end{eqnarray*}
We have if $\mathrm{Pf}K(1,1)\geq 0$, $g'(t)>0$, $g(t)$ is strictly
increasing as $t$ increases. Since the duality transformation changes the sign of $\mathrm{Pf}K(1,1)$, we have $g(0)<0$, see (\ref{dualitysignchange}), and $\lim_{t\rightarrow\infty}g(t)>0$. When $t$ increases from 0 to $\infty$, originally $g(t)$ is negative, and there $g(t)=0$ is possible for some $t$. Consider the first such $t$ that $g(t)=0$. From then on, $g(t)$ will always be positive since $g'(t)>0$. As a result, there is a unique $t_0$,
such that $g(t_0)=0$. This proves the lemma.
\end{proof}

\noindent \textbf{Proof of Theorem 1.1} Theorem 1.1 follows directly
from Lemmas 4.11, 4.12 and 4.13.

\begin{example}($1\times 2$ Ising model) Consider an Ising model on
a square grid whose interactions have period $1\times 2$, as
illustrated in Figure 5.
\begin{figure}[htbp]
  \centering
\includegraphics*{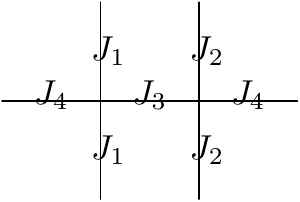}\qquad\rotatebox{40}{\includegraphics*{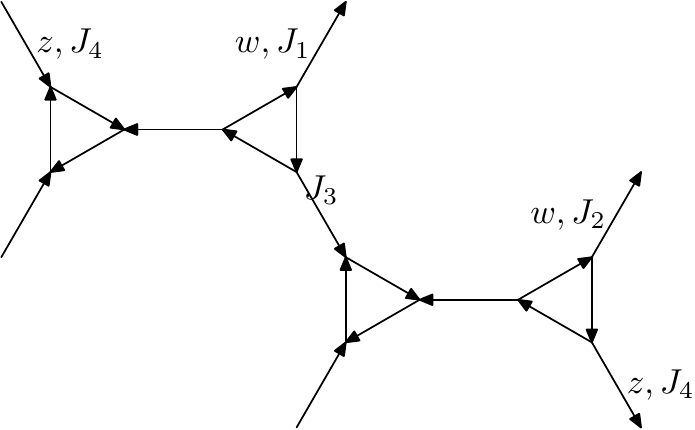}}
   \caption{$1\times 2$ Ising model and Fisher graph}
\end{figure}
The critical reciprocal temperature $\beta_c$ is the solution of the
following equation:
\begin{eqnarray*}
P(z=1,w=-1)&=&\det K(1,-1)\\
&=&{\left[b_2+c_1b_2c_2-b_1b_2+c_1c_2b_1b_2-1+c_2c_1+b_1+c_1b_1c_2\right]^2}=0,
\end{eqnarray*}
where
\begin{eqnarray*}
b_1=\tanh(\beta_cJ_4),\qquad b_2=\tanh(\beta_c J_3),\\
c_1=\tanh(\beta_cJ_1),\qquad c_2=\tanh(\beta_c J_2).
\end{eqnarray*}
In this case, $w=-1$, because the size of the period is odd$\times$even.
\end{example}

\section{Simple Proof for the Exponential Decay of Spin
Correlations with Symmetric Interactions}

The goal of this section is to prove the exponential decay of
spin-spin correlations for symmetric interactions above the critical
temperature. Namely
\begin{eqnarray*}
\langle\sigma_i\sigma_j\rangle\leq e^{-\alpha|i-j|},
\end{eqnarray*}
where $\alpha>0$ is a constant. The symmetric interactions are those who are invariant under the reflexion by the horizontal axis $y=0$ and the vertical axis $x=0$, see Figure 6.
\begin{figure}[htbp]
  \centering
\includegraphics{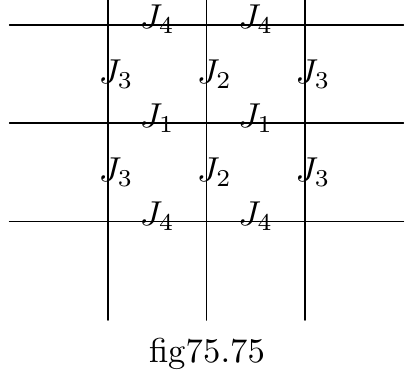}
   \caption{$1\times 2$ Ising model and Fisher graph}
\end{figure}
Lemma 5.1 proves the exponential decay along the lattice direction,
and Lemma 5.2 proves that the spin-spin correlations along arbitrary
directions are dominated by those along the lattice direction. Hence
the result follows.

\begin{lemma} Under the assumption that $\det K(z,w)\neq 0$
\[\langle\sigma_{00}\sigma_{0N}\rangle^2-\lim_{N\rightarrow\infty}\langle\sigma_{00}\sigma_{0N}\rangle^2\leq Ke^{-\alpha N},\]
where $0<K,\alpha<\infty$ are positive constants.
\end{lemma}
\begin{proof} According to (\ref{spincorrelation}), it suffices to consider the asymptotic behavior of $\prod_{0\leq l\leq N-1}(1-\tau_{e_l}^2)^2\det
T_N[\psi]$, as $N$ large. This is the determinant of a block Toeplitz matrix, obtained from $T_N$ by multiplying the $j$th row ($1\leq j\leq N$) by $(1-\tau_{e_{j-1}}^2)^2$. The determinant of its
symbol is $1$, according to Lemma 4.5, since $\tau_{e_j}$ changes periodically.  In the proof of this lemma, we denote the new
Toeplitz matrix by the same notation $T_N$, and the new symbol by
the same notation $\psi$, which is the original $\psi$ multiplied by
a constant for each row. Assume
\begin{eqnarray*}
\tilde{\psi}&=&\psi(z^{-1}),\\
T[\psi]&=&\lim_{n\rightarrow\infty}T_n[\psi],\\
T[\tilde{\psi}]&=&\lim_{n\rightarrow\infty}T_n[\tilde{\psi}].
\end{eqnarray*}
Define $A_r\bigcap K_r$ the Banach algebra under the norm
\begin{eqnarray*}
\|\phi\|=\sum_{k=-\infty}^{\infty}\|\phi_k\|+\{\sum_{k=-\infty}^{\infty}|k|\|\phi_k\|^2\}^{\frac{1}{2}},
\end{eqnarray*}
where $\phi$ is an $r\times r$ matrix-valued function. $\|\cdot\|$
on the right side is any norm for $r\times r$ matrices, and $\phi_k$
is the $k$th Fourier coefficient of $\phi$. First consider the case
that the Toeplitz operator $T[\tilde{\psi}]$ is invertible, which is
equivalent to the matrix-valued function $\psi$ has a factorization
\[\psi=\psi_{+}\psi_{-},\]
where $\psi_{\pm}$ are invertible in $A_r\bigcap K_r$, and
$\psi_{+}$ and $\psi_{-}$ have Fourier coefficients that
vanish for negative, resp. positive indices. The Hankel operator
$H[\psi]$ is Hilbert-Schmidt, the product of any two such is of
trace class.
\begin{eqnarray}
H[\psi]
&=&(\psi_{i+j+1}),\ \ \ \ \ 0\leq i,j\leq \infty,\\
T_n[\psi\phi]-T_n[\psi]T_n[\phi]&=&P_nH[\psi]H[\tilde{\phi}]P_n+Q_nH[
\tilde{\psi}]H[\phi]Q_n, \label{toeplitzhankel}
\end{eqnarray}
where $P_n$ and $Q_n$ are defined by
\begin{eqnarray*}
P_n(f_0,f_1,...)=(f_0,...,f_n,0,...,0),\\
Q_n(f_0,f_1,...)=(f_n,...,f_0,0,...,0).
\end{eqnarray*}
If we replace $\phi$ by $\psi_{-}^{-1}$ in  (\ref{toeplitzhankel}),
since $H[\psi_{-}^{-1}]=0$, we have
\[T_n[\psi_{+}]-T_n[\psi]T_n[\psi_{-}^{-1}]=P_nH[\psi]H[\tilde{\psi}_{-}^{-1}]P_n.\]
Multiply the above equation by $T_n[\psi_{+}^{-1}]$, we have
\[T_n[\psi_{+}]T_n[\psi_{+}^{-1}]-T_n[\psi]T_n[\psi_{-}^{-1}]T_n[\psi_{+}^{-1}]=P_nH[\psi]H[\tilde{\psi}_{-}^{-1}]P_nT_n[\psi_{+}^{-1}].\]
Use $\psi_{+},\psi_{+}^{-1}$ to substitute $\psi,\phi$ in
(\ref{toeplitzhankel}), we have
\[I_n-T_n[\psi_{+}]T_n[\psi_{+}^{-1}]=P_nH[\psi_{+}]H[\tilde{\psi}_{+}^{-1}]P_n+Q_nH[\tilde{\psi}_{+}]H[\psi_{+}^{-1}]Q_n\]
Since $H[\tilde{\psi}_{+}^{-1}]=0$ and $H[\tilde{\psi}_{+}]=0$, we
have
\[T_n[\psi_{+}]T_n[\psi_{+}^{-1}]=I_n.\]
Therefore
\begin{eqnarray*}
T_n[\psi]T_n[\psi_{-}^{-1}]T_n[\psi_{+}^{-1}]&=&I_n-P_nH[\psi]H[\tilde{\psi}_{-}^{-1}]P_nT_n[\psi_{+}^{-1}]\\
&=&P_n(I-H[\psi]H[\tilde{\psi}_{-}^{-1}]P_nT[\psi_{+}^{-1}])P_n.
\end{eqnarray*}
Since $T_n[\psi_{\pm}^{-1}]$ are block triangular matrices, one sees
that the left side has determinant exactly
\[D_n[\psi]G[\psi_{-}^{-1}]^{n+1}G[\psi_{+}^{-1}]^{n+1}=\frac{D_n[\psi]}{G[\psi]^{n+1}},\]
where $D_n[\psi]$ is the determinant of $T_n[\psi]$, and $G$ is defined as in Lemma 4.4. Moreover, explicit computations show that
\[H[\psi]H[\tilde{\psi}_{-}^{-1}]T[\psi_{+}^{-1}]=H[\psi]H[\tilde{\psi}^{-1}].\]
Let $\|\cdot\|_{1}$ denote the norm of operators in
$\mathfrak{F}_1$(see Lemma 4.6, and the appendix). The uniform norm $|\cdot|$ of an operator $A$  is the number
\begin{eqnarray*}
|A|=\sup_{|\phi|=1}|A\phi|.
\end{eqnarray*}
We have
\begin{eqnarray*}
&&\|P_nH[\psi]H[\tilde{\psi}_{-}^{-1}]P_nT[\psi_{+}^{-1}]P_n-H[\psi]H[\tilde{\psi}^{-1}]\|_{1}\\
&\leq&\|P_nH[\psi]H[\tilde{\psi}_{-}^{-1}]P_nT[\psi_{+}^{-1}]P_n-P_nH[\psi]H[\tilde{\psi}_{-}^{-1}]T[\psi_{+}^{-1}]P_n\|_1\\
&&+\|P_nH[\psi]H[\tilde{\psi}^{-1}]P_n-H[\psi]H[\tilde{\psi}^{-1}]\|_1\\
&\leq&
Ke^{-\alpha n},
\end{eqnarray*}
where $0<K,\alpha<\infty$ are positive constants. To see the last bound, note that $H[\tilde{\psi}_{-}^{-1}]P_n-H[\tilde{\psi}_{-}^{-1}]$ has the form $\left(\begin{array}{cc}0&*\\0&*\end{array}\right)$, where the nonzero entries correspond to columns of $H[\tilde{\psi}_{-}^{-1}]$ with indices greater than $n$, which are exponentially small by definition, similarly for $P_nH[\tilde{\psi}]-H[\tilde{\psi}]$  .Define
\begin{eqnarray*}
A_n&=&P_nH[\psi]H[\tilde{\psi}_{-}^{-1}]P_nT[\psi_{+}^{-1}]P_n,\\
A&=&H[\psi]H[\tilde{\psi}^{-1}],
\end{eqnarray*}
then
\[\|A_n-A\|_1\leq Ke^{-\alpha n}.\]

Define
\begin{eqnarray*}
D_{A_n}(\mu)&=&\det(I-\mu A_n),\\
D_{A}(\mu)&=&\det(I-\mu A),\\
A(\mu)&=&A(I-\mu A)^{-1},\\
A_n(\mu)&=&A_n(I-\mu A_n)^{-1}.
\end{eqnarray*}
We shall call the complex number $\mu$ an F-regular point of the
operator $A$, if the operator $I-\mu A$ has an inverse. Let $\Gamma$
be a simple rectifiable contour, consisting of F-regular points of
the operator $A$, which encloses the point $\mu=0$ and the point
$\mu=1$. By virtue of the maximum modulus principle, if we can prove
that for any $\|A_n-A\|_1<Ke^{-\alpha n}$,
\[\max_{\mu\in\Gamma}|D_A(\mu)-D_{A_n}(\mu)|<CKe^{-\alpha n},\]
then
\[|D_A(1)-D_{A_n}(1)|<CKe^{-\alpha n}.\]
Let us denote by $L$ a simple rectifiable curve consisting of
F-regular points of the operator $A$, which connects the point
$\mu=0$ with some point of the contour $\Gamma$. Let
$\Gamma_\mu(\mu\in \Gamma\bigcup L)$ denote the shortest path along
the curves $\Gamma$ and $L$ which connects the point $\mu=0$ to the
point $\mu$. From the equality
\[(I-\mu A_n)^{-1}=(I-\mu A)^{-1}[I-\mu(A_n-A)(I-\mu A)^{-1}]^{-1},\]
it follows that when the condition
\[\|A_n-A\|_1<\min_{\mu\in\Gamma\bigcup L}[|\mu||(I-\mu A)^{-1}|]^{-1}\]
is fulfilled, all points $\mu$ of the curves $\Gamma\bigcap L$ are
F-regular points of the operator $A_n$ and
\[\max_{\mu\in\Gamma\bigcup L}|(I-\mu A_n)^{-1}|<C\max_{\mu\in \Gamma\bigcup L}|(I-\mu A)^{-1}|,\]
where $C$ is a constant depending only on the curves $\Gamma$ and
$L$ and the operator $A$. Since
\[A(\mu)-A_n(\mu)=(I-\mu A)^{-1}(A-A_n)(I-\mu A_n)^{-1},\]
we have
\[\|A(\mu)-A_n(\mu)\|_1\leq|(I-\mu
A)^{-1}|\|A-A_n\|_1|(I-\mu A_n)^{-1}|.\] for any $\mu_0\in
\Gamma\bigcup L$. Then
\[\left|\int_{\Gamma_{\mu_0}}tr(A(\mu)-A_n(\mu))d\mu\right|\leq \int_{\Gamma_{\mu_0}}\|A(\mu)-A_n(\mu)\|_1
d\mu<C\|A-A_n\|_1,\] where $C$ is another constant depending only on
the curves $\Gamma$ and $L$ and the operator $A$, and $tr(K)$ is the
sum of all eigenvalues of operator $K$(for details, see the appendix). 
$\Gamma_{\mu_0}$ is the shortest path along the curves $\Gamma$ and $L$ connecting the points $0$ and $\mu_0$. Then
\[|D_A(\mu_0)-D_{A_n}(\mu_0)|=\left|D_A(\mu_0)\left[1-\frac{D_{A_n}(\mu_0)}{D_A(\mu_0)}\right]\right|,\]
and
\begin{eqnarray*}
\left(\log\frac{D_{A_n}(\mu)}{D_A(\mu)}\right)'&=&\frac{D_{A_n}'(\mu)}{D_{A_n}(\mu)}-\frac{D_{A}'(\mu)}{D_{A}(\mu)}\\
&=&-\sum_{j=1}^{r(A)}\frac{\lambda_j(A)}{1-\mu\lambda_j(A)}+\sum_{j=1}^{r(A_n)}\frac{\lambda_j(A_n)}{1-\mu\lambda_j(A_n)}\\
&=&tr(A(\mu)-A_n(\mu)).
\end{eqnarray*}
where $r(K)$ is the dimension of operator $K$, namely, the dimension
of the closure of the range of $K$. Hence
\begin{eqnarray*}
|D_A(\mu)-D_{A_n}(\mu)|
=|D_A(\mu_0)\{1-\exp[\int_{\Gamma_{\mu_0}}tr(A(\mu)-A_n(\mu))]\}|<C\|A-A_n\|_1,
\end{eqnarray*}
where $C$ is another constant depending only on the curves $\Gamma$
and $L$ and the operator $A$,  then we have
\[|D_{A_n}(1)-D_{A}(1)|<CKe^{-\alpha n}.\]

Now we consider the case $T[\tilde{\psi}]$ is not invertible.
According to \cite{wi}, there is a $\phi$ with only finitely many
non-vanishing Fourier coefficients such that
$T[\tilde{\psi}+\varepsilon\tilde{\phi}]$ is invertible for all
sufficiently small nonzero $\varepsilon$. Let
$\eta_{\varepsilon}=\psi+\varepsilon\phi$, then the previous process
shows that
\[\|P_nH[\eta_\varepsilon]H[\tilde{\eta_\varepsilon}_{-}^{-1}]P_nT[{\eta_\varepsilon}_{+}^{-1}]P_n-H[\eta_\varepsilon]H[\tilde{\eta_\varepsilon}^{-1}]\|_{1}<Ke^{-\alpha n},\]
where $K,\alpha$ are independent of $\varepsilon$ since $\phi$ has
only finitely many non-vanishing Fourier coefficients. We consider
$\varepsilon$ on the boundary of a disk $D$
\[\max_{\varepsilon\in\partial D}|D_{A_n,\varepsilon}(1)-D_{A,\varepsilon}(1)|<CKe^{-\alpha n},\]
where
\begin{eqnarray*}
A_{\varepsilon}&=&H[\eta_{\varepsilon}]H[\tilde{\eta}_{\varepsilon}^{-1}],\\
A_{n,\varepsilon}&=&P_nH[\eta_\varepsilon]H[\tilde{\eta_\varepsilon}_{-}^{-1}]P_nT[{\eta_\varepsilon}_{+}^{-1}]P_n.
\end{eqnarray*}
Then the maximal modulus theorem says that
\[|D_{A_n}(1)-D_{A}(1)|<CKe^{-\alpha n}.\]
Since $G[\psi]=1$, we have
\[|\det T_n[\psi]-\lim_{n\rightarrow\infty}\det T[\psi]|< Ke^{-\alpha n},\]
where $0<K,\alpha<\infty$ are constants.
\end{proof}

\begin{lemma}Assume each fundamental domain has symmetric edge
weights with respect to a center vertex, as illustrated in Figure 8.
Assume $(p_0,q_0)$ is a vertex on the boundary of a fundamental
domain centered at $(0,0)$. If $i$ is even,
\[\langle\sigma_{p_0q_0}\sigma_{p_0+im,q_0+jn}\rangle_p\leq
\langle\sigma_{p_0q_0}\sigma_{p_0+im,q0}\rangle_p.\]
\end{lemma}
\begin{proof}Let $W$ be a layer of cylinder consisting of $t$
fundamental domains. In other words, if we use the number of
fundamental domains as a measure for the size of $W$, then the
circumference of $W$ is $t$, and the width of $W$ is 1. The transfer
matrix $Q_t$ is defined to be a square matrix whose rows (columns)
are labeled by the upper(lower) boundary configurations of $W$.
Namely, let $X$ be an upper boundary configuration, and $Y$ be an
lower boundary configuration, then
\begin{eqnarray*}
Q_t(X,Y)&=&\sum_{\{\sigma_r|r\in W\setminus\partial
W\}}\exp[-\beta(\frac{1}{2}\sum_{\{a,b|a\sim b,\sigma_a,\sigma_b\in
X\}}J_{ab}\sigma_a\sigma_b+\frac{1}{2}\sum_{\{c,d|c\sim
d,\sigma_c,\sigma_d\in Y\}}J_{cd}\sigma_c\sigma_d\\
&&+\sum_{e\in E(W\setminus\partial W)}J_e\sigma_u\sigma_v)].
\end{eqnarray*}
Without the external magnetic field,
\begin{eqnarray*}
\langle\sigma_{p_0q_0}\sigma_{p_0+im,q_0+jn}\rangle_p&=&\langle\sigma_{p_0q_0}\sigma_{p_0+im,q_0+jn}\rangle_p-\langle\sigma_{p_0q_0}\rangle_p\langle\sigma_{p_0+im,q_0+jn}\rangle_p\\
&=&4(\langle\rho_{p_0q_0}\rho_{p_0+im,q_0+jn}\rangle_p-\langle\rho_{p_0q_0}\rangle_p\langle\rho_{p_0+im,q_0+jn}\rangle_p).
\end{eqnarray*}
Hence it suffices to prove that
\[\langle\rho_{p_0q_0}\rho_{p_0+im,q_0+jn}\rangle_p-\langle\rho_{p_0q_0}\rangle_p\langle\rho_{p_0+im,q_0+jn}\rangle_p\leq\langle\rho_{p_0q_0}\rho_{p_0+im,q_0}\rangle_p-\langle\rho_{p_0q_0}\rangle_p\langle\rho_{p_0+im,q_0}\rangle_p.\]
Assume we have an $s\times t$ torus, where $s$ and $t$ are the
number of fundamental domains along each direction, then
\[\langle\rho_{p_0q_0}\rangle_{st}=\frac{\sum_{\{\sigma|\sigma_{p_0q_0}=1\}}\exp[-\beta\sum_{\{u,v|u\sim
v\}}J_{uv}\sigma_u\sigma_v]}{\sum_{\sigma}\exp[-\beta\sum_{\{u,v|u\sim
v\}}J_{uv}\sigma_u\sigma_v]}=\frac{tr(F_tQ_t^s)}{tr(Q_t^s)},\] where
$F_t$ is a diagonal matrix with
\[F_t(X,X)=\left\{\begin{array}{cc}1&\mathrm{if\ \sigma_{p_0q_0}=1\ in\ X}\\0&\mathrm{else}\end{array}.\right.\]
Since $Q_t$ is a symmetric matrix, all its eigenvalues are real,
there exists an orthogonal matrix $S_t$, such that
$\tilde{Q}_t=S^{-1}_tQ_tS_t$ is diagonal, and the modulus of its
eigenvalues decreases along the diagonal. Let
$\tilde{F}_t=S_t^{-1}F_tS_t$. Moreover, since $Q_t$ has strict
positive entries, the Perron-Frobenius Theorem says that the largest
eigenvalue in modulus of $Q_t$, $\lambda_{1,t}$, is simple and
strictly positive. Hence
\[\lim_{s\rightarrow\infty}\langle\rho_{p_0q_0}\rangle_{st}=\lim_{s\rightarrow\infty}\frac{tr(\tilde{F_t}\tilde{Q}_t^s)}{tr(Q_t^s)}=\tilde{F}_t(1,1).\]
Similarly,
\[\langle\rho_{p_0+im,q_0+jn}\rangle_{st}=tr(Q_t^{i}G_tQ_t^{s-i}),\]
and $G_t$ is a diagonal matrix with
\[G_t(X,X)=\left\{\begin{array}{cc}1&\mathrm{if \sigma_{p_0+im,q_0+jn}=1\ in\
X}\\0&\mathrm{else}\end{array}.\right.\] Assume
$\tilde{G}_t=S_t^{-1}G_tS_t$, then
\[\lim_{s\rightarrow\infty}\langle\rho_{p_0+im,q_0+jn}\rangle_{st}=\tilde {G}_t(1,1).\]
Finally we have
\begin{eqnarray*}
 \lim_{s\rightarrow\infty}\langle\rho_{p_0q_0}\rho_{p_0+im,q_0+jn}\rangle_{st}&=&\lim_{s\rightarrow\infty}\frac{tr(F_tQ_t^{i}G_tQ_t^{s-i})}{tr{Q_t^s}}\\
 &=&\frac{\tilde{F}_t\tilde{Q}_t^{i}\tilde{G}_t(1,1)}{\lambda_{1,t}^{i}}\\
 &=&\frac{\sum_k \tilde{F}_t(1,k)\lambda_{k,t}^{i}\tilde{G}_t(k,1)}{\lambda_{1,k}^{i}},
\end{eqnarray*}
therefore
\begin{eqnarray*}
&&\lim_{s\rightarrow\infty}(\langle\rho_{p_0q_0}\rho_{p_0+im,q_0+jn}\rangle_{st}-\langle\rho_{p_0q_0}\rangle_{st}\langle\rho_{p_0+im,q_0+jn}\rangle_{st})\\
&=&\frac{\sum_{k\neq1}\tilde{F}_t(1,k)\lambda_{k,t}^{i}{\tilde{G}_t(1,k)}}{\lambda_{1,k}^{i}}\\
&\leq&\frac{[\sum_{k\neq
1}\tilde{F}_t^2(1,k)\lambda_{k,t}^{i}\sum_{k\neq
1}\tilde{G}_t^2(1,k)\lambda_{k,t}^{i}]^{\frac{1}{2}}}{\lambda_{1,t}^{i}}\\
&=&\lim_{s\rightarrow\infty}\sqrt{(\langle\rho_{p_0q_0}\rho_{p_0+im,q_0}\rangle_{st}-\langle\rho_{p_0q_0}\rangle_{st}\langle\rho_{p_0+im,q_0}\rangle_{st})}\\
&&\sqrt{(\langle\rho_{p_0,q_0+jn}\rho_{p_0+im,q_0+jn}\rangle_{st}-\langle\rho_{p_0+im,q_0+jn}\rangle_{st}\langle\rho_{p_0,q_0+jn}\rangle_{st})}\\
&=&\lim_{s\rightarrow\infty}\langle\rho_{p_0q_0}\rho_{p_0+im,q0}\rangle_{st}-\langle\rho_{p_0q_0}\rangle_{st}\langle\rho_{p_0+im,q0}\rangle_{st}.
\end{eqnarray*}
The inequality follows from the Cauchy-Schwartz, it is an inner
product since $i$ is even, and the last equality follows from
translation invariance. The result follows by letting
$t\rightarrow\infty$, and the fact that the limit of the spin-spin
correlation is independent of the order of $s\rightarrow\infty$ and
$t\rightarrow\infty$.
\end{proof}

Assume we have a finite graph. Let $\gamma_{x,y}$ be a path
connecting $(0,0)$ and $(x,y)$, then
\begin{eqnarray*}
\langle\sigma_{00}\sigma_{xy}\rangle=Z^{-1}\prod_{e\in\gamma_{x,y}}\tanh
J_e\sum_{\sigma}\prod_{e\in\gamma_{x,y}}(1+\frac{1}{\tanh
J_e}\sigma_u\sigma_v)\prod_{e\in\gamma_{x,y}^c}(1+\tanh
J_e\sigma_u\sigma_v).
\end{eqnarray*}
The denominator (partition function $Z$) corresponds to the closed polygon configurations
while the numerator (all the other factors on the right) corresponds to configurations on the graph which
have an odd number of present edges at vertices $(0,0)$ and $(x,y)$,
and an even number of present edges at all the other vertices. To see that, first of all, the expansion$\sum_{\sigma}\prod_{e\in\gamma_{x,y}}(1+\frac{1}{\tanh
J_e}\sigma_u\sigma_v)\prod_{e\in\gamma_{x,y}^c}(1+\tanh
J_e\sigma_u\sigma_v)$ is the same as the weighted sum closed polygon configuration where all the edges except those on $\gamma_{xy}$ have weight $\tanh J_e$, while the edges on $\gamma_{xy}$ have weight $\frac{1}{\tanh{J_e}}$. When multiplied by $\prod_{e\in\gamma_{x,y}}\tanh
J_e$, this is the same as the weighted sum of configurations which can be obtained by a closed polygon configuration by changing all the present edge on $\gamma_{x,y}$ to be vacant, and all the vacant edges on $\gamma_{x,y}$ to be present, under the condition that all the edges in the graph have weight $\tanh_{J_e}$. For all the vertices which do not lie on $\gamma_{x,y}$ and in the interior of $\gamma_{x,y}$, there are still an even number of present adjacent edges after the change. However, for the two endpoints $(0,0)$ and $(x,y)$, an odd number of adjacent edges are present in the new configuration. The weighted sum of new configurations can be represented by a dimer model on graphs by removing an vertex of the gadgets corresponding to $(0,0)$ and $(x,y)$, where the original gadgets are illustrated as in Figure 2.
Namely, the 8 local configurations at a degree-4 vertex, with an even number of edges present, are in one-to-one correspondence with dimer configurations on two gadgets, each of which is a gadget form the right graph of Figure 2 by removing one vertex. Figure 7 and 8 give two examples of such correspondence.
\begin{figure}[htbp]
\centering
\includegraphics*{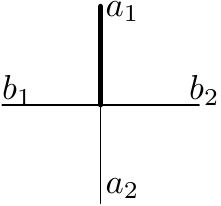}\qquad\qquad
\includegraphics*{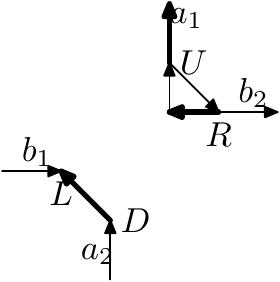}
\caption{odd-edge configuration and modified gadget: 1}
\end{figure}

\begin{figure}[htbp]
\centering
\includegraphics*{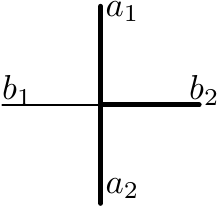}\qquad\qquad
\includegraphics*{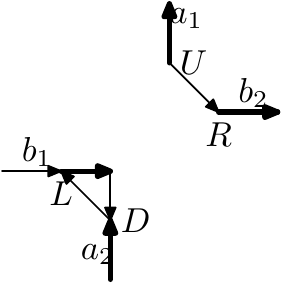}
\caption{odd-edge configuration and modified gadget: 2}
\end{figure}

Hence $\langle\sigma_{0,0}\sigma_{x,y}\rangle$ can be expressed as a sum
of monomer-monomer correlations the Fisher graph, that is, the
partition function of dimer configurations on a graph with two
vertices removed, divided by the partition function of the original
graph. 

We have the following corollary

\begin{corollary}Above the critical temperature, the monomer-monomer
correlation on symmetric, periodic, ferromagnetic Fisher graph
decays to zero exponentially fast. The decay rate is rotationally
invariant.
\end{corollary}

\appendix
\section{Toeplitz and Hankel Operators}
If $\phi$ is a matrix-valued function defined on the unit circle with Fourier coefficients $\phi_k$, then $T[\phi]$, $H[\phi]$ are, respectively, the semi-infinite Toeplitz and Hankel matrices,
\begin{eqnarray*}
T[\phi]&=&(\phi_{i-j}),\qquad 0\leq i,j< \infty,\\
H[\phi]&=&(\phi_{i+j+1}),\qquad 0\leq i,j<\infty.
\end{eqnarray*}
$\phi$ is called the symbol of the Toeplitz and Hankel matrices. If $\phi$ is bounded, these may be thought as operators from the Hilbert space $\ell_2$ to $\ell_2$. In addition we writie
\[\tilde{\phi}(z)=\phi(z^{-1}).\]
There is a simply identity relating Toeplitz and Hankel operators 
\begin{equation}
T[\phi\psi]-T[\phi]T[\psi]=H[\phi]H[\tilde{\psi}^{-1}]\label{th}.
\end{equation}
Identity (\ref{th}) is trivial. The left side has (i,j) block
\begin{eqnarray*}
\sum_{k=-\infty}^{\infty}\phi_{i-k}\psi_{k-j}-\sum_{k=0}^{\infty}\phi_{i-k}\psi_{k-j}=\sum_{k=-\infty}^{-1}\phi_{i-k}\psi_{k-j}=\sum_{k=0}^{\infty}\phi_{i+k+1}\psi_{-k-j-1}.
\end{eqnarray*}
Two applications of the identity (\ref{th}) give (I=identity matrix)
\begin{eqnarray*}
T[\phi]T[\phi^{-1}]&=&I-H[\phi]H[\tilde{\phi}^{-1}],\\
T[\phi^{-1}]T[\phi]&=&I-H[\phi^{-1}]H[\tilde{\phi}].
\end{eqnarray*}

Let $X$ and $Y$ be Banach spaces. A bounded linear operator $T: X\rightarrow Y$ is called completely continuous if, for any weakly convergent sequence $(x_n)$ from $X$, the sequence $Tx_n$ is norm-convergent in $Y$.

A compact operator $L$ from $X$ to $Y$ is a linear operator such that the image under $L$ of any bounded subset of $X$ is a relatively compact subset of $Y$. Compact operators on Banach spaces are always completely continuous.

The collection $\mathfrak{F}_p(1\leq p<\infty)$ consists of all completely continuous operators $A$ for which
\begin{eqnarray*}
\sum_{j=1}^{\infty}s_j^p(A)<\infty,
\end{eqnarray*} 
where $s_j$ are singular values of $A$. The norm in $\mathfrak{F}_p$ is defined by 
\begin{eqnarray*}
\|A\|_p=(\sum_{j=1}^{\infty}s_j^p(A))^{\frac{1}{p}}.
\end{eqnarray*}
When $p=1$ the operator is called a trace-class operator. Here is an equivalent definition for the trace-class operator over a separable Hilbert space. A bounded linear operator $A$ over a separable Hilbert space $H$ is said to be in the trace class if for some orthonormal basis $\{e_k\}_{k=1}^{\infty}$ of $H$, the sum
\begin{eqnarray*}
\|A\|_1=\sum_{k}<(A^*A)^{\frac{1}{2}}e_k,e_k>
\end{eqnarray*}
is finite, where $A^*$ is the adjoint operator of $A$ under the inner product in $H$. In this case
\begin{eqnarray*}
\sum_{k}<Ae_k,e_k>
\end{eqnarray*}
is absolutely convergent and is independent of the choice of orthonormal basis.  This limit is denoted by $tr(A)$. When $A\in\mathfrak{F}_1$, this limit is equal to the sum of all eigenvalues of the operator $A$.

A Hilbert-Schmidt operator is an operator on Hilbert space satisfying $\|A\|_2<\infty$. A trace-class operator on Hilbert space is always a Hilbert-Schmidt operator. A Hilbert-Schmidt operator is always compact.

Given a Hankel matrix with continuous symbol $\phi$. The Fourier coefficients of $\phi$ decays exponentially fast. In this case, the Hankel operator is a Hilbert-Schmidt, and the product of two Hankel operators is a trace-class operator. Since $T[\phi]T[\phi^{-1}]=I-H[\phi]H[\tilde{\phi}^{-1}]$, we have $T[\phi]T[\phi^{-1}]$ is an operator differing from the identity by an operator of trace class, hence the determinant $E$ in Widom's theorem (Lemma 4.4) is well defined.

\section{Uniqueness of Gibbs Measure on Dimer Models of Non-bipartite Graphs}

One correspondence between the Ising spin configurations and dimer configurations is described in \cite{cim}. Consider the dual graph $\mathcal{G}^{*}$ of the graph $\mathcal{G}$ where the original Ising model is defined. Replace each vertex of the dual graph by a gadget. 

\begin{figure}[htbp]
\centering
\includegraphics*{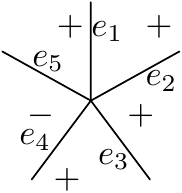}\qquad\qquad
\includegraphics*{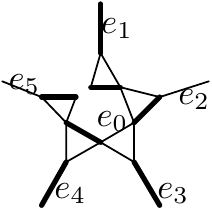}
\caption{Dual Ising Graph and Decorations}
\end{figure}

Assume all the edges which do not correspond to the edges of the dual graph have weight 1. 
For those edges corresponding to edges of the dual graph, assume an even number of such edges, surrounding each gadget, have weight less than 1, all the other dual edges have weights greater than or equal to 1. For example, in Figure 9, edges $e_1,e_2,e_3,e_4,e_5$ have corresponding edges in the dual graph(dual edges), while the edge $e_0$ does not correspond to any edge in the dual graph at all.
The two-to-one correspondence between the Ising spin configurations and the dimer configurations is described as follows:

\begin{enumerate}
\item If two adjacent spins have the same sign, and the weight of the edge
of $\mathcal{G}^{*}$ separating the two spins is bigger than 1, then the
edge  is present in the
dimer configuration;
\item If two adjacent spins have the same sign,
and the weight of the edge of $\mathcal{G}^{*}$ separating the two spins is less
than 1, then the  edge  is
not present in the dimer configuration;
\item If two adjacent spins have the opposite sign,
and the weight of the edge of $\mathcal{G}^{*}$ separating the two spins is less
than 1, then the edge  is
present in the dimer configuration;
\item If two adjacent spins have the opposite sign,
and the weight of the edge of $\mathcal{G}^*$ separating the two spins is bigger
than 1, then the  edge  is
not present in the dimer configuration;
\end{enumerate}

For example, in Figure 9, assume the weights satisfy  $w_{e_2}<1$, $w_{e_4}<1$, and $w_{e_1},w_{e_3},w_{e_5}>1$, then an Ising spin configuration on the left graph should correspond to a dimer configuration on the right graph, according to the 4 criteria listed above. 

\begin{corollary}For a graph obtained from the dual graph of a ferromagnetic Ising model as described above, there is a unique translation-invariant Gibbs measure defined on dimers.
\end{corollary}
\begin{proof} It suffices to prove the result on any finite cylindrical set. Let $\rho_e$ be the variable associated to an edge $e$, that is, if $e$ is present in the dimer configuration, $\rho_e=1$, otherwise $\rho_{e}=0$. If $e$ is not an dual edge, the configuration of $e$ is uniquely determined by dual edges. For example, in Figure
\[\rho_{e_0}=\rho_{e_3}\rho_{e_4}(1-\rho_{e_5})+\rho_{e_3}(1-\rho_{e_4})\rho_{e_5}\]
Hence it suffices to prove the result on an arbitrary cylindrical set consisting of dual edges. When $e$ is a dual edge
 let $\sigma_u,\sigma_v$ be
spins on endpoints of the dual edge $e^*$, then
\begin{eqnarray*}
\rho_e=\left\{\begin{array}{cc}\frac{1-\sigma_u\sigma_v}{2}&\mathrm{if}\
w_e<1\\ \frac{1+\sigma_u\sigma_v}{2}&\mathrm{if}\
w_e>1\end{array}\right.
\end{eqnarray*}
Hence $\langle\rho_{e_1}\cdots\rho_{e_k}\rangle$ depends only on spin-spin even correlation functions. The uniqueness of spin-spin even correlation functions implies the uniqueness of the translation-invariant Gibbs measure on dimers. 
\end{proof}


\begin{thebibliography}{99}

\bibitem{ai2}M.~Aizenman, Rigorous Results of Critical Behavior II,
Statistical physics and dynamical systems (K\"{o}szeg, 1984),
453--481, Progr. Phys., 10, Birkhäuser Boston, Boston, MA, 1985

\bibitem{abf}M.~Aizenman, D.~Barsky, R.~Fern\'{a}ndez, the Phase
Transition in a General Class of Ising-type Models is Sharp.
J.~Statist.~Phys., 47(1987), no.3-4,343-374

\bibitem{ba}R.~Baxter, Exactly Solved Models in Statistical
Mechanics, Academic Press(1982)

\bibitem{bt}C.~Boutillier, B.~de Tili\`{e}re, the Critical
Z-invariant Ising Model via Dimers: the Periodic Case,
Probab.~Theory Relat.~Fields(2010)147:379-413

\bibitem{cim}D.~Cimasoni, a Generalized Kac-Ward Formula, 
J. Stat. Mech. (2010) P07023

\bibitem{ckp}H.~Cohn, R.~Kenyon, J.~Propp, a Variational Principle for
Domino Tilings, J. Amer. Math. Soc., 14(2001), No.2, 297-346

\bibitem{fi}M.~Fisher, On the Dimer Solution of Planar Ising Models.
Journal of Mathematical Physics, 7:1776-1781, October 1966

\bibitem{fkg}C.~Fortuin, P.~Kasteleyn and J.~Ginibre, Correlation
Inequalities on some Partially Ordered Sets, Comm.~Math.~Phys,
22(1971), No.~2, 89-103

\bibitem{gf}I.~C.~Gohberg, I.~A.~Feldman, Convolution Equations and
Projection Methods for their Solutions, AMS(1974)

\bibitem{gk}I.~C.~Gohberg, M.~G.~Krein, Introduction to the Theory
of Linear Nonselfadjoint Operators, AMS(1969)

\bibitem{gri}R.~Griffiths, Rigorous Results and Theorems, in Phase
Transition and Critical Phenomena, edited by C.~Domb and M.~Green,
New York: Academic Press 1972.

\bibitem{ho}R.~Holley, Remarks on the FKG Inequalities,
Comm.~Math.~Phys, 36(1974), No.3, 227-231

\bibitem{ka1}P.~Kasteleyn, the Statistics of Dimers on a Lattice,
Physica, 27(1961), 1209-1225

\bibitem{ka2}P.~Kasteleyn, Graph Theory and Crystal Physics, in
Graph Theory and Theoretical Physics, Academic Press, London, 1967

\bibitem{ke1}R.~Kenyon, an Introduction to the Dimer Model, arxiv:
math/0310326

\bibitem{ke2}R.~Kenyon, Local Statistics on Lattice Dimers, Ann.
Inst. H. Poicar$\acute{e}$. Probabilit$\acute{e}$s, 33(1997),
591-618

\bibitem{kos}R.~Kenyon, A.~Okounkov, S.~Sheffield, Dimers and Amoebae,
Ann. Math. 163(2006), no.3, 1019-1056

\bibitem{ko}R.~Kenyon, A.Okounkov, Planar Dimers and Harnack Curve,
Duke. Math. J. 131(2006), no.3, 499-524

\bibitem{lan}O.~Lanford III, Entropy and Equilibrium States in
Classical Statistical Mechanics, in Statistical Mechanics and
Mathematical Problems, A.~Lenard, ed., Springer(1973)

\bibitem{le}J.~Lebowitz, Bounds on the Correlations and Analyticity
Properties of Ferromagnetic Ising Spin Systems, Comm.~Math.~Phys,
28(1972), No.4

\bibitem{le2}J.~Lebowitz, Coexistence of Phases in Ising
Ferromagnets, J.~Statist.~Phys., 16(1977), no. 6, 453--461

\bibitem{lm}J.~Lebowitz, A.~Martin-L\"{o}f, On the Uniqueness of the
Equilibrium State for Ising Spin Systems, Comm.~Math.~Phys,
25(1972), No.4, 276-282

\bibitem{ly}T.~D.~Lee, C.~N.~Yang, Statistical Theory of Equations
of State and Phase Transitions.~II.~Lattice Gas and Ising Model,
Phys.~Rev, 87(1952), 410-419

\bibitem{li}Z.~Li, Spectral Curve of Periodic Fisher Graphs

\bibitem{mw}B.~McCoy and T.~Wu, the Two Dimensional Ising Model,
Harvard University Press,1973

\bibitem{sh}S.~Sheffield, Random Surfaces, Arterisque, No.~304(2005)

\bibitem{tes}G.~Tesler, Matchings in Graphs on non-orientable
surfaces, J.Combin.Theory, Ser.B. 78(2000), no.2, 198-231


\bibitem{wi}H.~Widom, Asymptotic Behavior of Block Toeplitz Matrices
and Determinants. II, Adv.~Math., 21(1976), 1-29


\end{thebibliography}
\end{document}